\documentclass[11pt]{article}

\usepackage[margin=1in]{geometry}

\usepackage{amsmath,amssymb,amsthm,mathtools,bm}

\usepackage{booktabs}
\usepackage{multirow}
\usepackage{array}
\usepackage{tabularx}
\usepackage{graphicx}
\usepackage{calc}
\usepackage{xcolor}

\usepackage{siunitx}
\usepackage{enumitem}
\usepackage{url}

\usepackage{tikz}
\usepackage{pgfplots}
\pgfplotsset{compat=1.18}

\usetikzlibrary{
    arrows.meta,
    positioning,
    fit,
    calc,
    decorations.pathreplacing,
    shapes.geometric
}

\newdimen\xfigwd
\newtheorem{definition}{Definition}
\newtheorem{proposition}{Proposition}

\newtheorem{corollary}{Corollary}

\newcommand{\R}{\mathbb{R}}

\newcommand{\Xspace}{\mathcal{X}}

\newcommand{\Aset}{\mathcal{A}}
\newcommand{\Fset}{\mathcal{F}}

\newcommand{\E}{\mathbb{E}}

\newcommand{\Var}{\operatorname{Var}}
\newcommand{\Cov}{\operatorname{Cov}}
\newcommand{\corr}{\operatorname{corr}}

\newcommand{\MI}{\mathcal{I}}

\newcommand{\PLV}{\mathrm{PLV}}
\newcommand{\PLI}{\mathrm{PLI}}
\newcommand{\wPLI}{\mathrm{wPLI}}

\title{A Mathematical Framework for EEG Functional Connectivity:
Taxonomy, Measure Relationships, Identifiability, and Analytical Bounds}

\author{
Jerrin Thomas Panachakel\\[3pt]
\small School of Electrical and Electronic Engineering\\
\small Technological University Dublin, Dublin, Ireland\\
\small \texttt{jerrin.panachakel@tudublin.ie}
}

\date{}

\begin{document}

\maketitle

\begin{abstract}
Electroencephalographic (EEG) functional connectivity is represented
by a large family of statistics that quantify linear dependence,
spectral coupling, phase synchronization, amplitude-envelope
association, information transfer, directed influence, cross-frequency
interaction, or graph topology. The proliferation of these measures
creates an under-formalized inferential question: given a reported
value of one connectivity measure, what can be deduced about another?
This paper develops a mathematical framework in which an EEG
connectivity estimator is treated as a mapping from an explicitly
defined admissible state space to an observable space. For a source
connectivity functional $C_i$ and a target functional $C_j$, the set
of target values compatible with an observed source value is defined
as a cross-connectivity attainable set. The framework yields precise
criteria for exact determination, conditional determination after
state augmentation, analytical bounded inference, approximate
inference, and non-identifiability. The formulation is extended to
multichannel connectivity matrices, graph measures, and dynamic
connectivity trajectories. Particular attention is given to the fact
that admissibility is experiment-dependent: filtering, referencing,
volume conduction, source leakage, inverse modeling, epoch duration,
stationarity assumptions, frequency selection, and estimator
regularization alter the state-to-observable map. Analytical case
studies show that covariance determines correlation after marginal
scale information is supplied, Gaussian mutual information is exactly
determined by correlation only within the jointly Gaussian family,
constant-amplitude narrow-band coherence reduces to PLV, while PLI
and wPLI are generally non-identifiable from one another without
further assumptions. The paper synthesizes foundational and recent
EEG/MEG connectivity literature and converts recurring methodological
cautions into explicit mathematical objects and propositions. The
resulting framework is not a ranking of connectivity estimators; it
is a language for stating, with assumptions made visible, what one
connectivity result does and does not identify about another.
\end{abstract}


\medskip

\section{Introduction}
\label{sec:introduction}
Functional connectivity is commonly used to quantify statistical dependence among distributed neural signals, whereas effective connectivity additionally addresses directed interactions or model-based influence \cite{Friston1994,Horwitz2003}. In electroencephalography (EEG), connectivity analysis is particularly attractive because millisecond-scale temporal resolution permits study of interactions across a wide range of temporal and oscillatory scales \cite{Varela2001,BuzsakiDraguhn2004}. The resulting literature is large and heterogeneous, spanning classical correlation and coherence, phase-based synchronization, envelope-based coupling, information-theoretic dependence, directed time-series methods, cross-frequency coupling, source-space connectivity, and graph-theoretic analysis \cite{Pereda2005,BastosSchoffelen2016,Sakkalis2011}.

The methodological diversity is scientifically useful but mathematically consequential. A connectivity statistic is a compression of an underlying signal configuration. Correlation preserves a normalized second-order relationship but discards absolute scale. Coherence preserves a frequency-specific linear dependence but not the full joint probability law. PLV preserves circular concentration of relative phase, while PLI and wPLI impose different transformations on the same phase-lag information \cite{Lachaux1999,Aydore2013,Stam2007,Vinck2010,Vinck2011,panachakel2023csp,panachakel2021classification,sidharth2023emotion}. Mutual information is sensitive to general statistical dependence, but a single correlation coefficient does not in general determine it \cite{Kraskov2004,Paninski2003,Reshef2011}. Directional measures introduce additional dependence on temporal embedding, model structure, conditioning variables, or spectral decomposition \cite{Granger1969,Geweke1982,Kaminski1991,KaminskiBlinowska2001,BaccalaSameshima2001,Dhamala2008a,Dhamala2008b,BarnettSeth2015}. Cross-frequency measures require another layer of specification because phase and amplitude estimates depend on band definitions, filtering, waveform shape, and stationarity assumptions \cite{CanoltyKnight2010,Tort2010,Aru2015,Hyafil2017}.

Several influential reviews already provide strong descriptions of individual estimators, families of estimators, and common interpretational pitfalls \cite{Pereda2005,LeocaniComi2006,BastosSchoffelen2016,Sakkalis2011,Chiarion2023}. Recent systematic reviews and perspectives further document the expanding use of EEG connectivity in disorders and connectomics, including epilepsy, depression, dementia, and acute brain injury, as well as the increasing emphasis on dynamic and network-level analyses \cite{Pegg2020,Miljevic2023,Khaleghi2024,Paitel2025,Anzalone2026,Coelli2025,ZhangChen2025}. Yet one mathematical question remains insufficiently formalized across these literatures: \emph{what does a value of one connectivity measure identify about another connectivity measure?}

The question is more demanding than asking whether two measures are correlated empirically. A deterministic relationship requires that the target measure be constant over every admissible state that produces the same source value. This is a statement about the fibers of two mappings. It is therefore distinct from an empirical regression, a high Pearson correlation between estimators, a common response to the same experimental manipulation, or agreement observed on a restricted benchmark dataset.

The central idea of this paper is to formulate cross-connectivity inference as an attainable-set problem. Let $\Xspace$ denote an explicitly specified set of admissible EEG states and let $C_i,C_j$ be source and target connectivity functionals. For an observed source value $c_i$, define the source fiber
\begin{equation}
\Fset_i(c_i)=\{x\in\Xspace:C_i(x)=c_i\}
\label{eq:fiber_intro}
\end{equation}
and the target attainable set
\begin{equation}
\Aset_{j|i}(c_i)=\{C_j(x):x\in\Xspace,\ C_i(x)=c_i\}.
\label{eq:attainable_intro}
\end{equation}
If this set is a singleton, the target is exactly determined. If it is bounded but not a singleton, the source yields an interval or more general bounded set. If it is unbounded, no finite universal target bound follows from the source observation under the selected admissible state space.

This framework has several advantages for EEG. Measurement and nuisance operations need not remain hidden in prose: the reference operator, forward mixing, filtering, source reconstruction, analysis window, and estimator parameters may all be represented explicitly in the state-to-observable map. The same mathematical definition can then describe sensor-space and source-space connectivity, static and dynamic connectivity, pairwise and graph-derived measures. The framework also clarifies why common statements such as ``measure $A$ is equivalent to measure $B$'' are meaningful only after the assumptions defining the admissible state space have been stated. Figure~\ref{fig:attainable_geometry} gives the corresponding geometric view.

The paper makes six contributions. First, it formalizes EEG connectivity measures as functionals and the associated equivalence classes they induce on signal space. Second, it introduces the cross-connectivity attainable set and a directional hierarchy of exact, conditional, bounded, approximate, and non-identifiable relationships. Third, it provides analytical results and constructive counterexamples for major connectivity families. Fourth, it treats preprocessing, referencing, volume conduction, inverse solutions, and leakage correction as transformations of the admissible state space rather than merely practical details. Fifth, it extends the framework to dynamic connectivity, connectivity matrices, and graph metrics. Sixth, it translates the mathematical analysis into a reporting and study-design protocol intended to make cross-measure claims reproducible and falsifiable.

The manuscript is structured as a Topical Review. Section~II defines terminology, scope, and notation. Section~III develops the attainable-set and identifiability framework. Section~IV maps the information retained and discarded by major EEG connectivity families. Sections~V--VIII derive representative relationships among linear, spectral, phase-based, information-theoretic, directed, and cross-frequency measures. Section~IX incorporates measurement physics, referencing, filtering, and source leakage. Sections~X and XI extend the framework to dynamic and graph-derived connectivity. Section~XII addresses finite-sample estimation and uncertainty. Section~XIII presents worked analytical examples. Section~XIV proposes a unified decision framework, Section~XV critically synthesizes the literature, Section~XVI provides reporting standards, Section~XVII states limitations, Section~XVIII develops future directions, and Section~XIX concludes the paper.

\section{Terminology, Scope, and Mathematical Notation}
\label{sec:scope}

\subsection{What is being treated as ``connectivity''?}
Throughout this paper, functional connectivity is used in its statistical sense: a quantified dependence between signals, transformed signal features, reconstructed sources, or latent representations. A nonzero functional-connectivity statistic is not interpreted as proof of an anatomical connection or a causal mechanism. This distinction is especially important for scalp EEG, for which multiple distributed sources contribute to each measured channel \cite{Friston1994,Michel2004,SchoffelenGross2009,Holler2017}.

A connectivity object may be pairwise or multivariate, undirected or directed, static or dynamic, sensor-space or source-space, and scalar, matrix-valued, or graph-valued. The mathematical framework does not require a particular estimator. Instead, each estimator is represented by a map from an admissible state to an observable.

\subsection{Research questions}
The mathematical review is organized around five questions. \textbf{RQ1:} when does one EEG connectivity observable exactly determine another over a declared signal class? \textbf{RQ2:} when exact determination fails, what additional assumptions or measured variables are sufficient to make the target identifiable? \textbf{RQ3:} what analytical bounds can be obtained without replacing a many-to-one mapping by an unjustified point conversion? \textbf{RQ4:} how do acquisition, reference, filtering, source reconstruction, leakage correction, finite sampling, and dynamic-window choices alter the attainable set? \textbf{RQ5:} how should cross-measure claims be reported so that their mathematical strength is reproducible and falsifiable?

These questions deliberately separate three layers that are often conflated in application papers: the \emph{population functional} being estimated, the \emph{measurement/processing operator} that produces the data supplied to that functional, and the \emph{estimator} used on a finite sample. Two studies can therefore use the same named connectivity measure while operating on different effective functionals because their references, filters, windows, source models, or normalization conventions differ. The principal measure families and their inferential roles are summarized in Table~\ref{tab:taxonomy}.

\subsection{Literature-curation strategy}
The paper is a mathematically driven topical review rather than a conventional clinical systematic review. Literature was curated around five evidence layers: (i) foundational mathematical definitions; (ii) original estimator papers; (iii) methodological and tutorial reviews; (iv) empirical validation, reliability, artifact, reference, and leakage studies; and (v) recent application-oriented reviews and perspectives. Priority was given to traceable journal articles or established conference papers with stable bibliographic identifiers. The final manuscript deliberately does not treat a large raw reference count as a surrogate for evidence quality. A larger bibliographic corpus is useful for literature discovery, but only sources that substantively support a claim should be cited in the main text.

\subsection{Scope}
The review covers covariance and correlation; partial correlation; coherency, coherence, and partial coherence; imaginary coherency and related lagged measures; PLV, PPC, PLI, and wPLI; amplitude-envelope connectivity; mutual information and entropy-based dependence; Granger causality, DTF, PDC, and transfer entropy; phase-amplitude and related cross-frequency coupling; multivariate and higher-order interactions; dynamic connectivity; graph/network measures; and the acquisition, reference, filtering, source-reconstruction, and leakage factors that change the resulting observables.

\begin{table*}[t]
\caption{Taxonomy of connectivity objects considered in the framework.}
\label{tab:taxonomy}
\centering
\scriptsize
\begin{tabularx}{\textwidth}{>{\raggedright\arraybackslash}p{2.3cm} >{\raggedright\arraybackslash}p{2.5cm} X X X}
\toprule
Family & Representative measures & Primary object & Information emphasized & Principal assumptions/confounds \\
\midrule
Time-domain linear & covariance, Pearson $r$, partial correlation & second-order moments & linear dependence and normalized co-fluctuation & scale, preprocessing, common reference \\
Spectral & coherency, coherence, partial coherence & cross-spectrum / auto-spectra & frequency-specific linear dependence & stationarity, tapering, spectral leakage, volume conduction \\
Phase-based & PLV, PPC, PLI, wPLI & relative-phase distribution & circular concentration, sign or magnitude-weighted lag & phase extraction, amplitude weighting, mixing \\
Amplitude-envelope & AEC, leakage-corrected AEC & band-limited amplitudes & slow envelope co-fluctuation & leakage, orthogonalization, band definition \\
Information-theoretic & MI, conditional MI, TE & probability distributions & nonlinear and/or directional statistical dependence & finite-sample bias, embedding, estimator choice \\
Directed linear & GC, DTF, PDC, PSI & temporal prediction / spectral transfer & directionality and frequency-resolved flow & model order, hidden drivers, filtering \\
Cross-frequency & PAC, phase-phase CFC, bispectral measures & joint statistics across bands & interaction between oscillatory scales & filtering, harmonics, waveform asymmetry, nonstationarity \\
Graph-derived & degree, strength, clustering, path length, modularity, centrality & matrix $\mathbf{W}$ & network organization after edge compression & thresholding, density, weights, node definition \\
Dynamic & $C(t)$, $\mathbf{W}(t)$, state sequences & trajectories / distributions over time & temporal variability, states, transitions & window length, state model, temporal dependence \\
\bottomrule
\end{tabularx}
\end{table*}

\subsection{Notation}
Let $N$ denote the number of EEG channels and $T$ the number of samples in an analysis window. The observed data matrix is $\mathbf{Y}\in\R^{N\times T}$. A latent source process is denoted by $\mathbf{S}$. A broad measurement model is
\begin{equation}
\mathbf{Y}=\mathcal{P}\,\mathcal{R}\,\mathcal{M}(\mathbf{S},\theta_m)+\mathbf{E},
\label{eq:measurement_model}
\end{equation}
where $\mathcal{M}$ is a forward or mixing operator, $\mathcal{R}$ is a reference transformation, $\mathcal{P}$ collects preprocessing operations, $\theta_m$ denotes acquisition and model parameters, and $\mathbf{E}$ contains noise and unmodeled contributions. The operators may be linear or nonlinear, fixed or data-dependent.

A connectivity functional is denoted by $C_k$. It may be scalar-valued, vector-valued, or matrix-valued. Unless stated otherwise, identifiability refers to the population or idealized functional, not to a particular finite-sample estimator. Section~XII explicitly restores estimation uncertainty.

\section{Cross-Connectivity Attainable Sets and Identifiability}
\label{sec:framework}

\subsection{Admissible state spaces}
\begin{definition}[Admissible EEG state space]
Let $\Xspace$ be the set of all signal configurations, probability laws, nuisance variables, and methodological parameters regarded as physically, statistically, and analytically admissible for the question under study. The definition of $\Xspace$ is part of the scientific claim.
\end{definition}

Examples of state restrictions include a fixed number of channels; bounded signal variance; jointly Gaussian signals; stationary narrow-band processes; a specified frequency range; a fixed reference; a known forward model; or bounded source-amplitude and noise parameters. Different choices produce different attainable sets.

\subsection{Fibers and attainable sets}
For a source functional $C_i$, the fiber associated with $c_i$ is
\begin{equation}
\Fset_i(c_i)=\{x\in\Xspace:C_i(x)=c_i\}.
\label{eq:fiber}
\end{equation}
The target attainable set is
\begin{equation}
\Aset_{j|i}(c_i)=\{C_j(x):x\in\Fset_i(c_i)\}.
\label{eq:attainable}
\end{equation}
The construction is directional even though no causal interpretation is implied. In general, $\Aset_{j|i}(c)$ and $\Aset_{i|j}(c)$ have different structure because each functional discards different information.

\begin{definition}[Exact determination]
$C_j$ is exactly determined by $C_i$ on $\Xspace$ if, for every admissible $c_i$, the set $\Aset_{j|i}(c_i)$ is a singleton.
\end{definition}

\begin{definition}[Conditional determination]
$C_j$ is conditionally determined by $C_i$ given an augmentation $Z=z$ if the augmented attainable set
\begin{equation}
\begin{aligned}
\Aset_{j|i,z}(c_i,z)
={}&\{C_j(x): C_i(x)=c_i,\\
&\qquad Z(x)=z,\ x\in\Xspace\}.
\end{aligned}
\end{equation}
is a singleton for every admissible pair $(c_i,z)$.
\end{definition}

\begin{definition}[Analytical bound]
If $\Aset_{j|i}(c_i)$ is bounded, define
\begin{equation}
L_{j|i}(c_i)=\inf \Aset_{j|i}(c_i),\qquad
U_{j|i}(c_i)=\sup \Aset_{j|i}(c_i).
\end{equation}
If the attainable set is compact, the extrema are attained and may be written as a minimum and maximum.
\end{definition}

\begin{definition}[Approximate identifiability]
For a metric $d$ on the target space, the source value $c_i$ is $\varepsilon$-identifying for $C_j$ if
\begin{equation}
\sup_{a,b\in\Aset_{j|i}(c_i)} d(a,b)\leq \varepsilon.
\end{equation}
The quantity on the left is the attainable-set diameter.
\end{definition}

\subsection{Partition and factorization view}
Each functional partitions $\Xspace$ into equivalence classes
\begin{equation}
x\sim_i x'\quad\Longleftrightarrow\quad C_i(x)=C_i(x').
\end{equation}
Exact inference from $C_i$ to $C_j$ is possible precisely when every equivalence class induced by $C_i$ is contained within an equivalence class induced by $C_j$.

\begin{proposition}[Fiber criterion for exact inference]
For functionals $C_i$ and $C_j$ on $\Xspace$, the following are equivalent:
\begin{enumerate}[leftmargin=*]
\item $C_j$ is exactly determined by $C_i$ on $\Xspace$;
\item $C_j$ is constant on every fiber of $C_i$;
\item there exists a function $g$ on the image of $C_i$ such that
\begin{equation}
C_j(x)=g(C_i(x)),\qquad \forall x\in\Xspace.
\end{equation}
\end{enumerate}
\end{proposition}

\noindent\textit{Proof sketch:} If (3) holds, all states with the same source value have the same target value. Conversely, if the target is constant on each source fiber, define $g(c)$ to be that common target value. The partition formulation is equivalent.

This criterion is stronger than an empirical association. A regression model $\widehat C_j=g(C_i)$ learned from data can achieve high predictive accuracy while exact determination fails because the conditional variance remains nonzero outside the sampled distribution. Deterministic identifiability is therefore a property of the map and the admissible state space, not of a particular dataset.

\begin{proposition}[Information order]
Define $C_i\succeq C_j$ when there exists a map $g$ such that $C_j=g\circ C_i$ on $\Xspace$. Then $\succeq$ is a preorder on connectivity functionals. If both $C_i\succeq C_j$ and $C_j\succeq C_i$, the two functionals induce the same partition of $\Xspace$ up to an invertible reparameterization on their images.
\end{proposition}

This partial ordering is useful because it reframes ``equivalence'' as information preservation. A measure need not be numerically identical to another measure to be at least as informative for a particular inferential target.

\subsection{Non-identifiability by constructive counterexample}
To prove non-identifiability, it is sufficient to construct two admissible states $x_1,x_2$ satisfying
\begin{equation}
C_i(x_1)=C_i(x_2),\qquad C_j(x_1)\neq C_j(x_2).
\end{equation}
This approach is particularly powerful for phase and information-theoretic measures because simple discrete phase distributions or mixture models can often realize the required equal-source/different-target pair exactly.

\subsection{Bounds and optimization}
When exact determination fails but the admissible state space contains explicit constraints, cross-measure inference becomes an optimization problem:
\begin{align}
L_{j|i}(c)&=\inf_{x\in\Xspace} && C_j(x)\\
\text{s.t.}&&& C_i(x)=c,
\label{eq:boundoptL}
\end{align}
and analogously for $U_{j|i}$. Closed forms are available for simple parameterizations. In more complex settings, numerical optimization, convex relaxations, moment methods, or interval analysis may be used, but the mathematical object remains the attainable set itself.

\begin{figure*}[t]
\centering
\begin{tikzpicture}[
    >=Latex,
    font=\small,
    process/.style={
        draw,
        rounded corners=3pt,
        align=center,
        minimum height=0.9cm,
        inner sep=5pt
    },
    info/.style={
        draw,
        align=center,
        inner sep=6pt
    }
]


\node[process, minimum width=2.7cm] (state)
    at (0,0)
    {Admissible EEG state\\$x\in\Xspace$};

\node[process, minimum width=2.9cm] (source)
    at (3.5,0)
    {Source functional\\$c_i=C_i(x)$};

\node[process, minimum width=3.0cm] (fiber)
    at (7.7,0)
    {Source fiber\\$\Fset_i(c_i)$};

\node[process, minimum width=3.0cm] (target)
    at (12.6,0)
    {Target values\\$\Aset_{j|i}(c_i)$};

\draw[->] (state.east) -- (source.west);
\draw[->] (source.east) -- (fiber.west);
\draw[->] (fiber.east) -- (target.west);


\node[
    info,
    minimum width=2.7cm,
    minimum height=0.9cm
] (augment)
    at (3.5,-2.3)
    {Add state information\\$Z=z$};

\node[
    info,
    align=left,
    text width=4.7cm,
    font=\footnotesize
] (classes)
    at (7.7,-2.3)
    {%
    \textbf{Classification of the attainable set}\\[2pt]
    Singleton: exact determination\\
    Bounded, multi-valued: analytical bounds\\
    Small diameter: approximate identification\\
    Unbounded: non-identifiability
    };


\draw[->, dashed]
    (augment.north)
    -- (3.5,-0.95)
    -- (7.0,-0.95)
    -- (7.0,-0.45);

\draw[->, dashed]
    (fiber.south)
    -- (classes.north);

\end{tikzpicture}

\caption{Cross-connectivity inference formulated as an attainable-set
problem. A reported source-measure value identifies a fiber of admissible
states rather than a unique state. The target measure is exactly
determined only when it is constant over that fiber. Additional state
information can restrict the fiber and enable conditional determination
or tighter bounds.}

\label{fig:attainable_geometry}
\end{figure*}
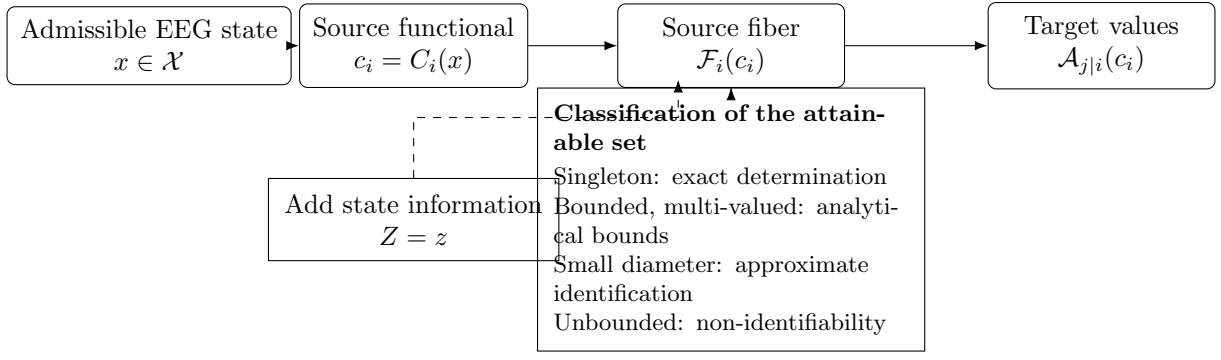

\subsection{Set-valued inverse inference}
The attainable-set formulation can be viewed as an inverse problem with a set-valued solution. Define the observation map $\Phi_i:\Xspace\rightarrow\mathcal C_i$ by $\Phi_i(x)=C_i(x)$. The inverse image of an observed value $c_i$ is the fiber $\Phi_i^{-1}(c_i)$. The target map $C_j$ applied to this inverse image produces the smallest exact set of target values compatible with the observation:
\begin{equation}
\Aset_{j|i}(c_i)=C_j\left(\Phi_i^{-1}(c_i)\right).
\end{equation}
This perspective is useful because it places no requirement on the inverse image to contain a unique state. The source statistic is allowed to be many-to-one, and the mathematical task is then to characterize the target image of the resulting fiber.

\begin{proposition}[Monotonicity under additional information]
Let $Z$ be an augmentation of the source observation. For any compatible $(c_i,z)$,
\begin{equation}
\Aset_{j|i,z}(c_i,z)\subseteq \Aset_{j|i}(c_i).
\end{equation}
Hence additional state information can never enlarge the attainable target set when it is conditioned on consistently.
\end{proposition}
\noindent\textit{Proof:} the augmented fiber is a subset of the unaugmented source fiber, and applying $C_j$ to a subset cannot produce a larger image. The result is elementary but important experimentally: measuring a quantity that reduces hidden degrees of freedom can turn a non-identifiable conversion into an exact or bounded one without changing the target estimator.

\begin{corollary}[Minimal sufficient augmentation]
If there exists a finite-dimensional augmentation $Z$ such that $C_j$ is constant on every joint fiber
\begin{equation}
\{x\in\Xspace:C_i(x)=c_i,\ Z(x)=z\},
\end{equation}
then $Z$ is sufficient for conditional determination of $C_j$ from $C_i$. Different sufficient augmentations need not be unique; the framework therefore supports an experimental-design interpretation in which the objective is to find the smallest measurable $Z$ that collapses the relevant fiber.
\end{corollary}

\subsection{Bounds as scientifically meaningful outputs}
A point conversion can be misleading when the target attainable set has non-negligible width. Define the normalized width for scalar measures by
\begin{equation}
W_{j|i}(c_i)=\frac{U_{j|i}(c_i)-L_{j|i}(c_i)}{1+|L_{j|i}(c_i)|+|U_{j|i}(c_i)|}.
\end{equation}
A small width indicates approximate identifiability on the selected state space; a large width indicates that the source observation leaves substantial target ambiguity. Unlike an $R^2$ from a regression, $W_{j|i}$ is a property of the admissible mapping rather than the frequency with which particular states happened to be sampled.

For matrix-valued measures, scalar lower and upper bounds are insufficient. Let $\mathbf C_i$ and $\mathbf C_j$ lie in normed spaces. The attainable set becomes
\begin{equation}
\Aset_{j|i}(\mathbf c_i)=\{\mathbf C_j(x):x\in\Xspace,\ \mathbf C_i(x)=\mathbf c_i\},
\end{equation}
and the natural uncertainty object is a region in matrix space. A useful family of summaries includes the diameter under Frobenius norm, elementwise envelopes, spectral-norm envelopes, or bounds on graph functionals $G(\mathbf C_j)$.

\begin{equation}
\operatorname{diam}_F(\Aset)=\sup_{A,B\in\Aset}\|A-B\|_F.
\end{equation}
The choice of metric should follow the scientific question rather than convenience: entrywise uncertainty is appropriate for edgewise inference, whereas spectral uncertainty is more relevant to stability or low-rank structure.

\subsection{A theorem on factorization through a summary statistic}
The fiber criterion can be expressed in a form familiar from sufficient-statistic theory. Let $S:\Xspace\rightarrow\mathcal S$ be any summary statistic and $T:\Xspace\rightarrow\mathcal T$ be a target functional. Then $T$ factors through $S$ if and only if there exists $g:\mathcal S\rightarrow\mathcal T$ with $T=g\circ S$. In the connectivity setting, the source connectivity value itself is the summary. The criterion therefore does not depend on the semantics of EEG; it follows from a general property of mappings.

The connection with statistical sufficiency is informative but should not be overstated. A source connectivity measure may be sufficient for a particular target connectivity functional without being sufficient for the underlying signal distribution, experimental parameter, or causal mechanism. The notion used here is therefore \emph{target-specific sufficiency}: sufficiency only for the requested cross-measure inference.

\subsection{Directionality of identifiability}
The notation $\Aset_{j|i}$ is intentionally directional. It is entirely possible that $C_j$ determines $C_i$ while $C_i$ does not determine $C_j$. A simple analogy is a many-to-one map $T=x^2$: knowing $T$ does not recover the sign of $x$, whereas knowing $x$ trivially determines $T$. In EEG, a target that retains more phase, amplitude, conditional, or multivariate structure can determine a compressed source statistic even when the reverse mapping is impossible. The symbols used in these definitions are collected in Table~\ref{tab:notation}.

\begin{table*}[t]
\caption{Mathematical objects used throughout the framework.}
\label{tab:notation}
\centering
\scriptsize
\begin{tabularx}{\textwidth}{p{2.5cm} p{3.1cm} X}
\toprule
Object & Definition & Interpretation in an EEG connectivity study \\
\midrule
$\Xspace$ & admissible state space & signal laws, nuisance parameters, processing choices, or constrained trajectories treated as possible \\
$C_i$ & source connectivity functional & reported connectivity statistic or matrix \\
$\Fset_i(c)$ & $\{x:C_i(x)=c\}$ & all admissible states producing the same reported source value \\
$\Aset_{j|i}(c)$ & $\{C_j(x):x\in\Fset_i(c)\}$ & all target values compatible with the source observation \\
$L,U$ & infimum/supremum of attainable target set & sharp scalar bounds when the target is not exactly identifiable \\
$\operatorname{diam}(\Aset)$ & maximum target-space distance inside the attainable set & measure of approximate identifiability \\
$Z$ & additional observed/assumed state information & e.g., variances, phase distribution, model order, reference, filter, or source geometry \\
$\widehat C_i$ & finite-sample estimator & actual quantity reported from an EEG recording \\
$\mathcal E_i$ & estimator error/uncertainty model & sampling variability, bias, or implementation uncertainty \\
\bottomrule
\end{tabularx}
\end{table*}

\begin{figure*}[t]
\centering
\begin{tikzpicture}[
    >=Latex,
    process/.style={
        draw,
        rounded corners=3pt,
        align=center,
        text width=2.55cm,
        minimum height=1.0cm,
        inner sep=4pt,
        font=\small
    },
    result/.style={
        draw,
        rounded corners=3pt,
        align=center,
        text width=2.4cm,
        minimum height=1.0cm,
        inner sep=4pt,
        font=\small
    },
    note/.style={
        draw,
        align=center,
        text width=2.65cm,
        inner sep=5pt,
        font=\scriptsize
    },
    classification/.style={
        draw,
        align=left,
        text width=3.1cm,
        inner sep=5pt,
        font=\scriptsize
    }
]


\node[process] (latent) at (0,0)
    {Latent neural state\\$\mathbf{s}(t)$};

\node[process] (measure) at (3.35,0)
    {Measurement and processing\\
     $\mathbf{y}=\mathcal{P}(\mathbf{s})$};

\node[process] (functional) at (6.70,0)
    {Connectivity functional\\
     $C_i=F_i(\mathbf{y})$};

\node[process] (fiber) at (10.05,0)
    {Inverse set\\$\Fset_i(c_i)$};

\draw[->] (latent.east) -- (measure.west);
\draw[->] (measure.east) -- (functional.west);
\draw[->] (functional.east) -- (fiber.west);


\node[note] (pipeline) at (3.35,-2.35)
    {\textbf{Pipeline factors}\\
     Reference, filtering, source model,\\
     window and artifact processing};

\node[classification] (classes) at (6.85,-2.35)
    {%
    \textbf{Inference outcome}\\[2pt]
    Singleton: exact determination\\
    Bounded, multi-valued: analytical bounds\\
    Small diameter: approximate inference\\
    Unbounded: non-identifiability
    };

\node[result] (target) at (10.05,-2.35)
    {Target attainable set\\
     $\Aset_{j|i}(c_i)$};

\draw[->] (fiber.south) -- (target.north);

\draw[->] (target.west) -- (classes.east);

\draw[->,dashed] (measure.south) -- (pipeline.north);

\draw[->,dashed] (pipeline.east) -- (classes.west);

\end{tikzpicture}

\caption{Expanded measurement-to-inference chain for EEG functional
connectivity. The connectivity functional operates on measured and
processed signals rather than directly on the latent neural state.
Consequently, the admissible state space and the resulting attainable
set depend on the measurement and processing pipeline, including
referencing, filtering, source reconstruction, epoch selection and
artifact handling. The attainable set determines whether cross-measure
inference is exact, bounded or non-identifiable under the specified
assumptions.}

\label{fig:measurement_chain}
\end{figure*}
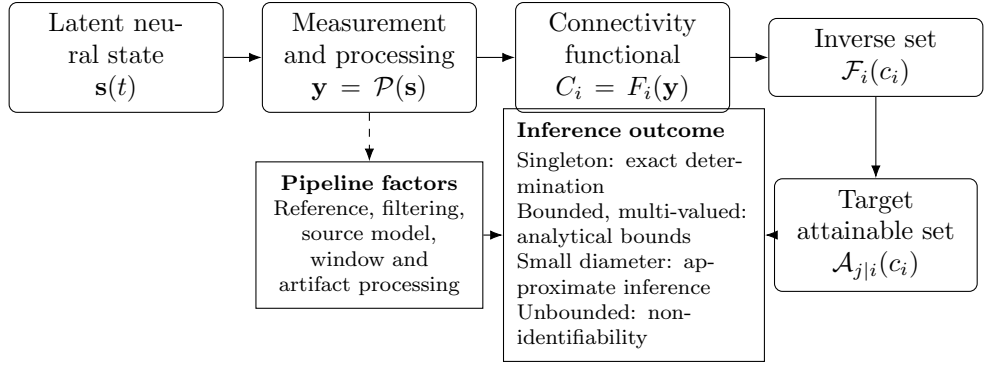

\section{What EEG Connectivity Measures Preserve and Discard}
\label{sec:taxonomy}
A useful way to compare connectivity estimators is to ask what summary object they retain. The same underlying signals can therefore be represented by different ``compressed'' objects, and cross-measure identifiability depends on whether one retained object is sufficient to reconstruct another. Table~\ref{tab:retention} makes this loss-of-information view explicit, while Table~\ref{tab:taxonomy2} organizes the same families by inferential object.

\begin{table*}[t]
\caption{Representative information retained and discarded by common connectivity families.}
\label{tab:retention}
\centering
\scriptsize
\begin{tabularx}{\textwidth}{p{2.4cm} p{3.2cm} X X}
\toprule
Measure family & Mathematical summary & Retained & Discarded or marginalized \\
\midrule
Covariance & $\E[(X-\mu_X)(Y-\mu_Y)]$ & signed second-order co-fluctuation & marginal scales, higher-order dependence, phase structure \\
Correlation & covariance normalized by $\sigma_X\sigma_Y$ & scale-free linear dependence & absolute amplitude scale, nonlinear structure \\
Coherency & $S_{xy}(f)/\sqrt{S_{xx}(f)S_{yy}(f)}$ & frequency-specific linear dependence and phase & full amplitude distribution, nonstationary structure outside the estimator \\
PLV & $|\E[e^{i\Delta\phi}]|$ & first circular moment magnitude & higher circular moments, amplitude weighting \\
PLI & $|\E[\operatorname{sign}(\sin\Delta\phi)]|$ & imbalance of phase-lead/lag signs & lag magnitude and fine angular structure \\
wPLI & $|\E[\Im(Z)|/\cdot\,\operatorname{sign}(\Im Z)]|$ & magnitude-weighted sign balance & unweighted sign balance and full amplitude/phase distribution \\
MI & $D_{KL}(p_{XY}\Vert p_Xp_Y)$ & general statistical dependence & temporal ordering unless extended to histories \\
GC & reduction in prediction-error variance & directional linear predictability under a model class & nonlinear dependence outside model, hidden variables not modeled \\
TE & conditional mutual information between present and past & directional information transfer in the chosen embedding & aspects outside selected histories and estimator support \\
Graph metric & $G(\mathbf W)$ & one network summary & most individual edge information \\
\bottomrule
\end{tabularx}
\end{table*}

The table shows why a numerical similarity between two measures is not sufficient to establish mathematical equivalence. Exact conversion requires one summary to preserve every feature used by the other.

\begin{table*}[t]
\caption{Taxonomy of EEG connectivity measures by inferential object rather than method name.}
\label{tab:taxonomy2}
\centering
\scriptsize
\begin{tabularx}{\textwidth}{p{2.0cm} p{2.0cm} p{2.0cm} p{3.0cm} X}
\toprule
Family & Pairwise/multivariate & Directedness & Typical state features emphasized & Representative literature \\
\midrule
Correlation/covariance & pairwise, extendable & undirected & time-domain second-order dependence & \cite{Friston1994,Nunez1997} \\
Coherence/coherency & pairwise/multivariate variants & undirected & frequency-specific second-order dependence & \cite{Nunez1997,Nolte2004} \\
Phase synchronization & pairwise & undirected & relative phase concentration or lag sign & \cite{Lachaux1999,Stam2007,Vinck2010,Vinck2011} \\
Envelope connectivity & pairwise/multivariate & usually undirected & amplitude co-fluctuation within bands & \cite{Brookes2012,SanchezBornot2018} \\
Information-theoretic & pairwise/multivariate & can be directional & full or conditional distributions & \cite{Schreiber2000,Kraskov2004,Vicente2011} \\
Autoregressive directionality & multivariate & directed & temporal prediction and transfer structure & \cite{Granger1969,KaminskiBlinowska2001,BarnettSeth2014} \\
Cross-frequency coupling & pairwise/multivariate & depends on functional & joint phase/amplitude structure across bands & \cite{Canolty2006,Tort2010,Aru2015,Hyafil2017} \\
Graph/network summaries & multivariate & directed or undirected & topology, community, integration, segregation & \cite{RubinovSporns2010,BassettSporns2017,Sporns2018} \\
\bottomrule
\end{tabularx}
\end{table*}

\section{Linear and Spectral Relationships}
\label{sec:linear}

\subsection{Covariance and correlation}
For random variables $X$ and $Y$ with finite nonzero variances,
\begin{equation}
r=\frac{\Cov(X,Y)}{\sqrt{\Var(X)\Var(Y)}}.
\label{eq:corr_cov}
\end{equation}
Hence covariance determines correlation when the two marginal standard deviations are available. Covariance alone does not. For any positive $a,b$, the transformed variables $X'=aX$ and $Y'=bY$ satisfy
\begin{equation}
\Cov(X',Y')=ab\Cov(X,Y),\qquad \corr(X',Y')=\corr(X,Y).
\end{equation}
Therefore a fixed correlation value can coexist with arbitrarily large covariance magnitude unless scale constraints are imposed.

\begin{proposition}[Conditional exactness of covariance-to-correlation conversion]
On an admissible state space with known standard deviations $(\sigma_X,\sigma_Y)$, correlation is exactly determined from covariance by $r=c/(\sigma_X\sigma_Y)$. On a state space in which the variances are free and unbounded, covariance-to-correlation conversion is not injective and correlation-to-covariance is unbounded.
\end{proposition}

\subsection{Correlation and Gaussian mutual information}
For a jointly Gaussian pair with correlation coefficient $\rho$,
\begin{equation}
\MI(X;Y)=-\frac{1}{2}\log(1-\rho^2).
\label{eq:gaussian_mi}
\end{equation}
The result follows from the Gaussian entropy formula and the determinant of the covariance matrix. It is an exact example of conditional determination: the target mutual information is a deterministic function of the source correlation only after joint Gaussianity is imposed \cite{Kraskov2004,Paninski2003}. The resulting population relationship is plotted in Fig.~\ref{fig:gaussian_mi}.

\begin{proposition}[Gaussian restriction is essential]
Over the class of all continuous bivariate distributions with fixed Pearson correlation $\rho$, mutual information is generally not uniquely determined by $\rho$. Consequently, \eqref{eq:gaussian_mi} is a model-conditional identity, not a universal correlation-to-MI conversion.
\end{proposition}

A constructive way to see the issue is to consider families formed by Gaussian copulas versus nonlinear deterministic or near-deterministic couplings adjusted to preserve the same Pearson correlation. The nonlinear structure can change mutual information substantially while leaving the second-order coefficient unchanged. This is the canonical distinction between a universal identity and a model-restricted identity.

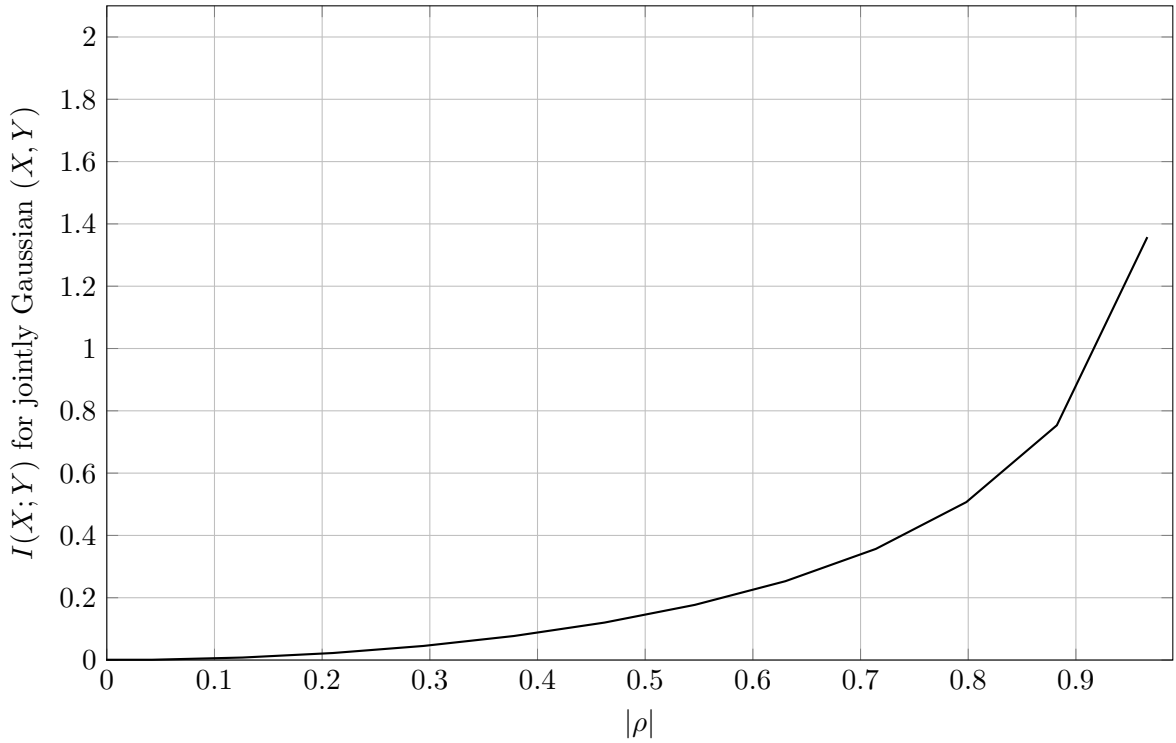
\begin{figure}[!t]
\centering
\begin{tikzpicture}
\begin{axis}[
    width=0.95\columnwidth,
    height=0.62\columnwidth,
    xlabel={$|\rho|$},
    ylabel={$I(X;Y)$ for jointly Gaussian $(X,Y)$},
    xmin=0,xmax=0.99,
    ymin=0,ymax=2.1,
    grid=major,
    samples=120]
\addplot[thick] {-0.5*ln(1-x^2)};
\end{axis}
\end{tikzpicture}
\caption{Exact Gaussian relationship between absolute correlation and mutual information. The curve is a population identity within the jointly Gaussian family; it is not a universal conversion for arbitrary EEG amplitude distributions.}
\label{fig:gaussian_mi}
\end{figure}

\subsection{Partial correlation and marginal correlation}
Partial correlation conditions on the remaining variables in a multivariate Gaussian model. For three variables,
\begin{equation}
r_{XY\cdot Z}=\frac{r_{XY}-r_{XZ}r_{YZ}}{\sqrt{(1-r_{XZ}^2)(1-r_{YZ}^2)}}.
\end{equation}
Thus marginal correlation alone cannot generally determine partial correlation because the cross-associations with the conditioning variable are discarded. The attainable-set view treats the omitted correlations as state variables in the optimization defining $\Aset_{\mathrm{partial}\,r|r_{XY}}$.

\subsection{Coherency and coherence}
The cross-spectrum $S_{xy}(f)$ and auto-spectra $S_{xx}(f),S_{yy}(f)$ define complex coherency
\begin{equation}
\gamma_{xy}(f)=\frac{S_{xy}(f)}{\sqrt{S_{xx}(f)S_{yy}(f)}},
\end{equation}
with coherence $|\gamma_{xy}(f)|^2$. Classical EEG coherency work emphasizes the effects of reference, volume conduction, and spatial scale on the observed spectrum \cite{Nunez1997,Nunez1997}. Spectral estimation itself is not neutral: the choice of window, taper, and smoothing alters the estimated cross-spectrum \cite{Welch1967,Thomson1982,Bruns2004}.

\subsection{Coherence, PLV, and the role of amplitude weighting}
Consider analytic signals
\begin{equation}
z_x(t)=A_x(t)e^{i\phi_x(t)},\qquad z_y(t)=A_y(t)e^{i\phi_y(t)}.
\end{equation}
A normalized cross-spectrum involves amplitude-weighted averages of $A_xA_y e^{i(\phi_x-\phi_y)}$, whereas PLV is
\begin{equation}
\PLV=\left|\E\left[e^{i\Delta\phi}\right]\right|,
\qquad \Delta\phi=\phi_x-\phi_y.
\end{equation}
If $A_x$ and $A_y$ are constant over the averaging interval and the signals are narrow-band under a coherent representation, the amplitude weights factor out and the magnitude of coherency reduces to PLV. This relationship has been analyzed explicitly in theoretical work \cite{LepageVijayan2017}.

\begin{proposition}[Constant-amplitude equivalence]
Under a narrow-band representation with constant amplitudes over the averaging interval, $|\gamma_{xy}|=\PLV$. If amplitude weights are variable, the equality need not hold.
\end{proposition}

The implication is important for EEG: a high empirical correlation between coherence and PLV in one dataset does not establish a universal identity. The amplitude and phase distributions constitute hidden state variables whenever the constant-amplitude restriction is not justified.

\subsection{Partial coherence and conditional spectral dependence}
Partial coherence is the frequency-domain analogue of conditioning in multivariate linear models. For a multivariate spectral density matrix $\mathbf S(f)$, conditional spectral quantities can be derived from the inverse spectral matrix. The key identifiability point is immediate: a marginal coherence between $X$ and $Y$ contains no complete information about their conditional coherence after a third or fourth process is controlled. The hidden cross-spectra with the conditioning variables are part of the fiber.

This is the spectral counterpart of the partial-correlation example. It also illustrates why multivariate connectivity should not be described merely as ``more variables in the same measure.'' Conditioning changes the target functional itself.

\subsection{Amplitude-envelope connectivity}
Let $a_x(t)$ and $a_y(t)$ denote analytic-signal envelopes in a specified band. Envelope correlation is often computed from a transformed version of these amplitudes. After an orthogonalization step designed to reduce zero-lag leakage, the observable becomes
\begin{equation}
C_{\mathrm{AEC}}=\corr\left(\mathcal O_x[a_x],\mathcal O_y[a_y]\right),
\end{equation}
where $\mathcal O$ may be a univariate or multivariate leakage-reduction operator. The operation is nonlinear with respect to the original raw signal pair because analytic extraction, orthogonalization, and correlation are composed operations. Consequently, no general conversion from phase connectivity to envelope connectivity exists without strong restrictions on amplitude--phase dependence.

\subsection{A second spectral identifiability result}
Consider two narrow-band signals with fixed coherence magnitude $q$. Their average phase difference and phase-difference concentration need not be unique because coherence combines magnitude and relative-phase information through a complex-valued cross-spectrum. A family of signal states can preserve $|\gamma_{xy}(f)|=q$ while changing the phase distribution across time/frequency. Thus a coherence magnitude generally does not determine PLV, PLI, or wPLI unless the relevant amplitude and phase model is restricted. The existence of this non-uniqueness is precisely why an empirical coherence--PLV correlation should not be reported as a universal conversion law.

\section{Phase-Based Connectivity}
\label{sec:phase}

\subsection{PLV as a circular moment}
The PLV is the magnitude of the first circular moment of the phase-difference distribution. It summarizes concentration around a preferred phase relation, but does not uniquely identify the underlying angular distribution. Two phase distributions can share the same first circular moment while differing in higher moments or sign balance.

A useful state-space representation is therefore
\begin{equation}
f_{\Delta\phi}(\theta),\qquad 0\leq \theta<2\pi,
\end{equation}
with
\begin{equation}
\PLV=\left|\int_0^{2\pi}e^{i\theta}f_{\Delta\phi}(\theta)\,d\theta\right|.
\end{equation}
This makes clear that PLV is one scalar functional of an entire probability density. The same principle underlies the theoretical treatment of PLV and PLI under circular distributions \cite{Aydore2013,Bruna2018}.

\subsection{PLI and wPLI}
The PLI emphasizes the sign of the imaginary part of the cross-spectrum or phase lag,
\begin{equation}
\PLI=\left|\E\left[\operatorname{sign}(\Im Z)\right]\right|,
\end{equation}
where $Z$ denotes an appropriate complex cross-spectral quantity. The weighted version is
\begin{equation}
\wPLI=\frac{|\E[\Im Z]|}{\E[|\Im Z|]},
\end{equation}
up to the estimator convention used in a particular implementation. PLI, wPLI, and PPC were developed to address distinct biases associated with phase synchronization and common sources \cite{Stam2007,Vinck2010,Vinck2011}.

\begin{proposition}[PLI does not identify wPLI in general]
There exist phase-lag distributions having the same PLI and different wPLI. Therefore PLI $\not\succeq$ wPLI on an unrestricted admissible phase-distribution space.
\end{proposition}

\noindent\textit{Constructive example:} Let the imaginary-part magnitudes be $m_+=0.5$ for positive-lag events and $m_-=0.5$ for negative-lag events, with probabilities $0.75$ and $0.25$, respectively. Then $\PLI=|0.75-0.25|=0.5$ and $\wPLI=0.5$. Keep the same sign probabilities but set $m_+=1$ and $m_-=0.1$. The PLI remains $0.5$, while
\begin{equation}
\wPLI=\frac{0.75(1)-0.25(0.1)}{0.75(1)+0.25(0.1)}\approx0.935.
\end{equation}
Hence the source scalar PLI cannot determine wPLI. The reverse implication also fails in the unrestricted setting because many combinations of sign probabilities and lag magnitudes can yield the same weighted ratio.

\subsection{PLI, wPLI, imaginary coherency, and PPC}
Imaginary coherency is another functional of complex cross-spectral quantities, while PLI and wPLI apply different nonlinear transformations before averaging. There is therefore no general one-to-one map among these measures. The distinction is not semantic: the measures can respond differently to amplitude variation, phase-lag magnitude, and zero-lag mixtures \cite{Nolte2004,Stam2007,Vinck2011}.

PPC is a pairwise-statistical correction designed to reduce sample-size bias in phase synchrony estimates \cite{Vinck2010}. Because PPC and PLV are related through the empirical distribution of pairwise phase differences but not numerically identical at finite sample size, neither should be treated as interchangeable without specifying the estimator and sample regime.

\subsection{Phase slope index}
The phase-slope index (PSI) estimates the sign of frequency-dependent phase slopes across neighboring frequencies. It targets directionality under assumptions different from those of PLI or wPLI \cite{Nolte2008}. A single scalar PLI therefore does not identify PSI, because PLI is insensitive to the frequency derivative of phase. The pairwise relationship classes considered in this section are summarized in Table~\ref{tab:relationships}, and the phase-oriented information dimensions are collected in Table~\ref{tab:phase_measures}.

\begin{table*}[t]
\caption{Representative pairwise relationships classified by attainable-set analysis.}
\label{tab:relationships}
\centering
\scriptsize
\begin{tabularx}{\textwidth}{p{2.2cm} p{2.2cm} p{2.8cm} X}
\toprule
Source & Target & Class & Missing information / required restriction \\
\midrule
Covariance & Correlation & Conditional exact & Marginal standard deviations \\
Correlation & Covariance & Unbounded without scale bounds & Absolute marginal scale \\
Gaussian correlation & Gaussian MI & Exact within Gaussian family & Joint Gaussianity \\
General correlation & MI & Non-identifiable & Higher-order/nonlinear dependence \\
Constant-amplitude coherency magnitude & PLV & Exact under signal model & Constant amplitude and compatible narrow-band representation \\
General coherence magnitude & PLV & Non-identifiable & Amplitude-phase weighting \\
PLI & wPLI & Non-identifiable & Magnitude distribution of phase lag \\
wPLI & PLI & Non-identifiable & Unweighted sign probabilities \\
PLI & PSI & Non-identifiable & Frequency-dependent phase slope \\
Gaussian TE & GC & Exact up to convention & Gaussian process/model assumptions \\
General TE & GC & Conditional / non-identifiable & Distribution, embedding, and model class \\
Single graph edge & network metric & Non-identifiable & Remaining edges, node definition, thresholding \\
Mean dynamic FC & variance/dwell time & Non-identifiable & Temporal trajectory \\
\bottomrule
\end{tabularx}
\end{table*}

\subsection{PLV, PPC, PLI, and wPLI as different projections of phase-lag structure}
For an epoch containing $N$ phase-lag observations $\Delta\phi_n$, the sample PLV is
\begin{equation}
\widehat{\PLV}=\left|\frac{1}{N}\sum_{n=1}^{N}e^{i\Delta\phi_n}\right|.
\end{equation}
The corresponding population quantity is the magnitude of the first circular moment. PPC instead forms a bias-reduced pairwise statistic and was introduced to avoid the sample-size dependence of conventional PLV 
\cite{Vinck2010}. PLI reduces each phase difference to the sign of its imaginary component:
\begin{equation}
\PLI=\left|\E[\operatorname{sign}(\sin\Delta\phi)]\right|,
\end{equation}
while wPLI weights those signs by the magnitude of the imaginary cross-spectrum:
\begin{equation}
\wPLI=\frac{|\E[\Im Z]|}{\E[|\Im Z|]},
\end{equation}
where $Z$ denotes the relevant complex cross-spectrum quantity 
\cite{Stam2007,Vinck2011}. These definitions make the discarded information explicit: PLI preserves sign imbalance; wPLI additionally retains a particular magnitude weighting; PLV retains angular concentration but is insensitive to the sign of the preferred lag.

The constructive PLI--wPLI counterexample above is sufficient to establish the non-identifiability result without introducing a second proof. In practical analyses, the important consequence is that the magnitude weighting used by wPLI must be treated as additional state information whenever one attempts a PLI-to-wPLI inference.

\subsection{Circular statistics and sample-size effects}
PLV is a circular concentration measure, not a raw coincidence count. Its expected finite-sample behavior depends on the number of observations, underlying phase concentration, and sampling scheme 
\cite{Aydore2013,Bruna2018}. This means that even when two estimators share a population-level quantity, their finite-sample estimators need not induce the same ordering across subjects or conditions. PPC was specifically developed to remove a sample-size bias present in pairwise phase consistency estimation 
\cite{Vinck2010}. Cross-measure inference must therefore state whether $C_i$ and $C_j$ refer to population functionals, biased sample estimators, or corrected estimators.

\begin{table*}[t]
\caption{Representative phase-based connectivity measures and the information dimensions emphasized by each functional.}
\label{tab:phase_measures}
\centering
\scriptsize
\begin{tabularx}{\textwidth}{p{1.5cm} p{4.1cm} p{3.0cm} X}
\toprule
Measure & Canonical summary & Main emphasis & Important discarded or conditioned information \\
\midrule
PLV & $|\E[e^{i\Delta\phi}]|$ & circular concentration & higher circular moments and signed lag direction \\
PPC & pairwise phase-vector similarity & bias-reduced phase consistency & direct amplitude magnitude; still phase-only \\
PLI & $|\E[\operatorname{sign}(\sin\Delta\phi)]|$ & lead/lag sign asymmetry & lag magnitude and angular distance from zero \\
wPLI & $|\E[\Im Z]|/\E[|\Im Z|]$ & magnitude-weighted nonzero-lag asymmetry & unweighted sign balance and detailed magnitude distribution \\
Imaginary coherence & normalized $|\Im S_{xy}|$ & nonzero phase-lag component of cross-spectrum & zero-lag component and full distribution \\
Phase-slope index & weighted frequency derivative of phase & directed propagation tendencies across frequency & instantaneous relation and many amplitude details \\
\bottomrule
\end{tabularx}
\end{table*}

\subsection{Zero-lag sensitivity is not the same as mechanistic specificity}
Measures designed to reduce zero-lag contamination should be described precisely. The imaginary part of coherency is insensitive to a purely real zero-phase cross-spectrum under the idealized instantaneous-mixing argument of Nolte et al. \cite{Nolte2004}. PLI and wPLI similarly target nonzero-lag phase relations 
\cite{Stam2007,Vinck2011}. However, a nonzero-lag relation can still arise from common driving, source mixing, filter-induced phase structure, or other indirect mechanisms. The attainable-set framework thus separates \emph{robustness to a nuisance} from \emph{identification of a biological mechanism}. These are different mathematical claims.

\section{Information-Theoretic and Directed Connectivity}
\label{sec:information}

\subsection{Mutual information}
Mutual information is
\begin{equation}
\MI(X;Y)=\int\!\!\int p(x,y)\log\frac{p(x,y)}{p(x)p(y)}\,dx\,dy.
\end{equation}
Its attainable set conditioned on correlation can be wide because the correlation coefficient is only one second-order summary. Estimation introduces a separate layer of uncertainty, with finite-sample bias depending on the estimator and support geometry \cite{Kraskov2004,Paninski2003}.

\subsection{Transfer entropy and Granger causality}
Transfer entropy can be written as a conditional mutual information between the future of one process and the past of another given the target's own past. For suitable Gaussian processes and matched conditioning conventions, transfer entropy and Granger causality are mathematically equivalent up to the choice of logarithmic base and scaling convention \cite{Schreiber2000,Barnett2009}. Outside that model class, the equivalence is lost.

\begin{proposition}[Model-conditional equivalence of Gaussian TE and GC]
For jointly Gaussian processes and a matched linear-history conditioning structure, the transfer-entropy functional and Granger-causality functional are deterministic transforms of one another. This does not imply a universal equivalence over arbitrary non-Gaussian processes.
\end{proposition}

The distinction mirrors the correlation/MI result: a relationship can be exact in a restricted state space and non-identifiable in a larger one. Consequently, the state-space restrictions should be reported whenever one connectivity measure is converted into another.

\subsection{DTF and PDC}
Directed transfer function (DTF) and partial directed coherence (PDC) are frequency-domain representations derived from multivariate autoregressive structure. DTF is naturally associated with normalized transfer magnitudes, whereas PDC normalizes differently and is related to the autoregressive coefficients themselves \cite{Kaminski1991,KaminskiBlinowska2001,BaccalaSameshima2001}. Granger causality provides an additional scalar or frequency-resolved representation of directional predictability. The three families are related but not interchangeable.

\subsection{Phase slope and spectral directionality}
The phase-slope index provides an alternative route to directional inference with a focus on frequency-dependent phase structure \cite{Nolte2008}. A key lesson from the framework is that a directional label alone does not define a common information object. DTF, PDC, GC, PSI, and TE use different aspects of the multivariate process. Cross-measure conversion therefore requires a model class and, often, additional state variables.

\subsection{Conditioning and hidden variables in directed connectivity}
For multivariate processes, pairwise directional measures can be confounded by an omitted common driver. Let $X\leftarrow Z\rightarrow Y$ denote a common-input configuration. A pairwise directional statistic between $X$ and $Y$ may be nonzero even when there is no direct interaction. Conditional or multivariate measures incorporate $Z$ into the model, but this changes the state dimension and therefore the fibers of the observable. The mathematical lesson is parallel to partial correlation: adding a conditioning variable can refine the inferential object, but it does not turn a pairwise value into a universal representation of the multivariate system.

\subsection{Multivariate autoregressive structure}
Suppose $\mathbf x_t$ follows a vector autoregressive model of order $p$,
\begin{equation}
\mathbf x_t=\sum_{k=1}^{p}\mathbf A_k\mathbf x_{t-k}+\boldsymbol\varepsilon_t.
\label{eq:var}
\end{equation}
The model parameters determine spectral transfer functions and several directed connectivity measures. However, a scalar GC summary is itself a compression of the coefficient matrices $\{\mathbf A_k\}$ and innovation covariance $\mathbf \Sigma_\varepsilon$. Two distinct VAR models can have the same scalar GC value for one direction but different PDC, DTF, spectral GC, or conditional GC elsewhere in the network. Hence directed-connectivity conversion is generally a matrix-to-scalar-to-matrix inverse problem.

\subsection{Frequency resolution and directionality}
Spectral directional measures depend on the frequency transformation of the fitted process. DTF and PDC differ in normalization and interpretation \cite{Kaminski1991,KaminskiBlinowska2001,BaccalaSameshima2001}. Fourier- and wavelet-based GC estimation provides additional choices in time--frequency localization \cite{Dhamala2008a,Dhamala2008b}. These choices alter the estimator family and therefore should appear explicitly in the admissible state description when cross-measure comparisons are attempted. Table~\ref{tab:directed} summarizes the resulting inferential distinctions among MI, TE, GC, DTF, PDC, and PSI.

\begin{table*}[t]
\caption{Representative information-theoretic and directed measures.}
\label{tab:directed}
\centering
\scriptsize
\begin{tabularx}{\textwidth}{p{2cm} p{4cm} p{3.0cm} X}
\toprule
Measure & Mathematical basis & Typical conditioning & Primary identifiability caveat \\
\midrule
MI & KL divergence $D_{KL}(p_{XY}\Vert p_Xp_Y)$ & none & depends on the full joint distribution \\
Transfer entropy & conditional MI between past/future variables & selected histories & embedding, estimator support, and hidden variables \\
GC & reduction in prediction error & target/source histories and model class & model order, linearity, and omitted variables \\
DTF & VAR transfer function magnitude & multivariate VAR & normalization and frequency-specific denominator \\
PDC & normalized AR coefficient representation & multivariate VAR & normalization changes the quantity represented \\
PSI & frequency-weighted phase-slope statistic & selected cross-spectra & depends on spectral resolution and phase structure \\
\bottomrule
\end{tabularx}
\end{table*}

\section{Cross-Frequency Connectivity}
\label{sec:cfc}
Cross-frequency coupling (CFC) describes interactions between oscillatory components operating at different frequencies. Phase-amplitude coupling (PAC), for example, can be expressed through a statistic of low-frequency phase and high-frequency amplitude. The literature has emphasized that apparent PAC can arise from waveform asymmetry, harmonics, filtering interactions, nonstationarity, and other sources \cite{CanoltyKnight2010,Tort2010,Aru2015,Hyafil2017}.

A generic PAC functional can be written
\begin{equation}
C_{\mathrm{PAC}}=F\!\left(\{\phi_{\ell}(t),A_h(t)\}_{t=1}^T\right),
\end{equation}
where the exact form of $F$ depends on the estimator. A single PAC value does not identify the joint time-frequency distribution, the amplitude marginal, or the low-frequency phase distribution. Two states can therefore have identical PAC but very different cross-frequency organization.

\begin{proposition}[PAC non-identifiability under unrestricted waveform classes]
If the admissible state space allows independent changes in waveform shape, amplitude distribution, or phase history while preserving the chosen PAC scalar, then the PAC value does not identify those latent features. Hence a PAC value cannot in general be translated into a unique cross-frequency phase relationship or a unique underlying waveform model.
\end{proposition}

The framework also clarifies why PAC estimators should not be compared solely by numerical magnitude. The estimators can induce different partitions of signal space, especially when they use different surrogates, phase estimators, modulation indices, or amplitude normalization schemes.

\subsection{A generic cross-frequency functional}
Let $\phi_\ell(t)$ be the phase of a low-frequency component and $A_h(t)$ the amplitude envelope of a high-frequency component. Many PAC methods can be written as a functional of their joint distribution,
\begin{equation}
C_{\mathrm{PAC}}=\mathcal F\!\left[p(\phi_\ell,A_h)\right].
\end{equation}
A scalar PAC value therefore discards the complete conditional amplitude distribution $p(A_h\mid\phi_\ell)$ and the phase distribution $p(\phi_\ell)$ except insofar as they influence the selected statistic. The same scalar can be obtained by different combinations of modulation depth, waveform asymmetry, amplitude variance, and phase occupancy.

\subsection{Waveform shape and harmonic confounding}
A non-sinusoidal oscillation contains harmonics. Filtering a non-sinusoidal waveform into nominally separated low- and high-frequency bands can therefore produce apparent phase--amplitude structure without a separate physiological oscillator. The literature on CFC has emphasized waveform shape, nonstationarity, and filtering as potential sources of spurious coupling \cite{CanoltyKnight2010,Tort2010,Aru2015,Hyafil2017}. In the present framework these are not merely ``pitfalls''; they correspond to alternative states $x\in\Xspace$ that can produce similar or identical reported PAC values while representing different generative mechanisms.

\begin{proposition}[Mechanism non-identifiability of a PAC scalar]
Let $M$ be a PAC statistic and suppose the admissible state space contains at least two generative families, one involving a genuine cross-frequency interaction and another involving waveform harmonics or filtering-induced dependence. If both families contain states with the same $M$ value, then the PAC scalar cannot identify the generative mechanism. A biological mechanism can therefore only be inferred after the admissible state space has been restricted to exclude the alternative explanations or additional measurements have been introduced.
\end{proposition}

\subsection{Surrogates test a null, not an identity}
Surrogate procedures are often used to test whether CFC exceeds a chosen null model. A rejection of a surrogate null does not establish that the observed PAC value uniquely determines any other connectivity metric, nor does it identify the physiological mechanism that generated the coupling. The result is a statement about incompatibility with a particular null distribution. This distinction is essential when a significant PAC result is subsequently translated into claims about synchronization or directed communication.

\section{Measurement Physics, Reference Choice, Filtering, and Source Leakage}
\label{sec:measurement}

\subsection{Measurement operator as part of the state-to-observable map}
Suppose latent sources satisfy $\mathbf{s}(t)\in\R^K$ and sensor data are
\begin{equation}
\mathbf{y}(t)=\mathbf{L}\mathbf{s}(t)+\boldsymbol\epsilon(t),
\end{equation}
where $\mathbf L$ is a forward/mixing matrix. Even when source signals are statistically independent, linear superposition can induce dependence at the sensors. This is the basic identifiability problem of the EEG inverse map \cite{Michel2004,Grech2008,SchoffelenGross2009}.

\subsection{Reference transformation}
For an $N$-channel record, an average-reference operator may be written
\begin{equation}
\mathbf R_{\mathrm{avg}}=\mathbf I-\frac{1}{N}\mathbf 1\mathbf 1^{\mathsf T}.
\end{equation}
The operator is idempotent but rank-deficient. Consequently it changes pairwise covariance and correlation structure and cannot be treated as a pure relabeling of channels. Reference choice has been shown empirically to affect information-theoretic and spatio-temporal interrelation measures \cite{Nunez1997,Nunez1997,Trujillo2017,Rios2019,Yao2001}.

\subsection{Volume conduction and instantaneous mixing}
A zero-lag instantaneous mixing model
\begin{equation}
\mathbf y=\mathbf M\mathbf s
\end{equation}
can create cross-channel dependence even when the source processes are uncoupled. Measures based on nonzero-lag or imaginary components were developed partly to reduce sensitivity to this problem \cite{Nolte2004,Stam2007,Vinck2011}. Importantly, reducing one confound does not make an estimator ``ground truth'': the transformation changes the admissible state-to-observable relationship rather than revealing the unobserved source relation exactly.

\subsection{Leakage correction and orthogonalization}
Source-space amplitude-envelope connectivity is sensitive to spatial leakage. Multivariate leakage correction and orthogonalization methods explicitly alter the source time series before connectivity is calculated \cite{Colclough2015,Brookes2012,SanchezBornot2018}. The mathematical consequence is straightforward: if $\mathcal O$ denotes an orthogonalization transform, then the relevant observable is $C(\mathcal O\mathbf s)$, not $C(\mathbf s)$. Two pipelines that use different leakage transforms are therefore operating on different admissible spaces.

\subsection{Filtering and phase estimation}
Digital filtering changes both amplitude and phase. Filter order, transition bandwidth, edge handling, and zero-phase versus causal implementation all affect the resulting analytic signal and hence connectivity \cite{Widmann2015,LeVanQuyen2001,Bruns2004}. High-pass filtering can additionally influence ICA-based artifact removal \cite{Widmann2015b,DelormeMakeig2004,PionTonachini2019}. In the attainable-set framework, filtering parameters are nuisance coordinates that must either be fixed, conditioned upon, or bounded. The major methodological transformations are summarized in Table~\ref{tab:transformations}, and the full measurement-to-functional chain is shown in Fig.~\ref{fig:measurement_chain}.

\begin{table*}[t]
\caption{Methodological transformations and their role in cross-measure identifiability.}
\label{tab:transformations}
\centering
\scriptsize
\begin{tabularx}{\textwidth}{p{2.4cm} p{2.8cm} X X}
\toprule
Transformation & Representative operator & Typical consequence & Identifiability implication \\
\midrule
Reference & $\mathbf R\mathbf Y$ & changes covariance geometry and rank & source/target relationships can change under re-reference \\
Filtering & convolution / frequency response & changes spectra, phase, transients & exact conversion requires fixed filter or explicit conditioning \\
ICA/BSS & component projection & removes or separates selected components & target may be altered by data-dependent decomposition \\
Leakage correction & orthogonalization / multivariate correction & suppresses instantaneous source similarity & connectivity before and after correction are different observables \\
Source reconstruction & inverse operator $\mathbf W$ & changes spatial representation & source connectivity does not uniquely identify sensor connectivity \\
Windowing & segment operator & changes stationarity regime and estimator variance & dynamic and static estimates cannot be interchanged without a window model \\
Thresholding & $W_{ij}\mapsto\mathbf 1(W_{ij}>\tau)$ & discards edge weights & graph measures can change discontinuously \\
\bottomrule
\end{tabularx}
\end{table*}

\subsection{Reference operators as rank-reducing maps}
The average-reference operator $\mathbf R_{\mathrm{avg}}$ satisfies
\begin{equation}
\mathbf R_{\mathrm{avg}}^2=\mathbf R_{\mathrm{avg}},
\qquad
\mathbf R_{\mathrm{avg}}\mathbf 1=\mathbf 0.
\end{equation}
Thus one dimension of the channel space is removed. Any connectivity quantity that depends on the covariance matrix is consequently evaluated on a lower-dimensional representation after average referencing. Other references impose different linear operators. There is no general reason for a source-target relationship derived under one reference to remain unchanged under another because both $C_i$ and $C_j$ may change through the transformation.

\subsection{Filtering as an operator on the admissible state}
Let $\mathcal H$ denote an LTI filter with transfer function $H(f)$. If $x(t)$ is replaced by $x_H(t)=\mathcal Hx(t)$, then its spectrum becomes $H(f)X(f)$ and its phase becomes $\arg H(f)+\arg X(f)$. In a causal or non-zero-phase implementation, group delay is another state variable. Even when the same nominal passband is reported, two filters with different transition bands can produce different edge transients and phase relations. The connectivity functional should therefore be indexed conceptually by the filter specification, e.g. $C_i^{(H)}$.

\subsection{Artifact rejection is part of the observation model}
Let $\mathcal A$ denote an artifact-removal operator. ICA, independent-component rejection, regression, signal-space projection, and related methods are data-dependent transformations. If $\mathcal A$ is estimated from the same dataset, the transformation itself is random and depends on the observed sample. The estimator can then be written as
\begin{equation}
\widehat C_i=C_i\!\left(\mathcal A(\mathbf Y;\widehat\theta_{\mathcal A})\right),
\end{equation}
where $\widehat\theta_{\mathcal A}$ is learned from the recording. This creates a second source of uncertainty in addition to the sampling variance of $C_i$. Standardized pipelines such as PREP and ICLabel make reproducibility easier to study, but they do not eliminate the need to report the actual processing choices \cite{BigdelyShamlo2015,PionTonachini2019,Delorme2012}.

\subsection{Forward mixing and a non-uniqueness proposition}
If $\mathbf y=\mathbf L\mathbf s$ and the forward operator $\mathbf L$ is non-injective (i.e., has a nontrivial null space), multiple source configurations are compatible with the same sensor observation. Even when a particular inverse solution $\widehat{\mathbf W}$ is chosen, $\widehat{\mathbf s}=\widehat{\mathbf W}\mathbf y$ is an estimator determined by the chosen forward model, regularization, and constraints \cite{Grech2008,Michel2004,SchoffelenGross2009,PascualMarqui1994,PascualMarqui2011}. Thus a sensor connectivity value generally does not identify a unique source connectivity value.

\begin{proposition}[Sensor-to-source non-identifiability]
Suppose there exist $x_1,x_2\in\Xspace$ with identical sensor-level observable $C_i^{\mathrm{sens}}(x_1)=C_i^{\mathrm{sens}}(x_2)$ but distinct source-level target values $C_j^{\mathrm{src}}(x_1)\neq C_j^{\mathrm{src}}(x_2)$. Then source connectivity is not identifiable from the sensor connectivity statistic alone. This remains true even if a deterministic source reconstruction algorithm is applied, unless the inverse operator and the admissible source state have been sufficiently restricted to collapse the relevant fiber.
\end{proposition}

\section{Dynamic Connectivity and Temporal Identifiability}
\label{sec:dynamic}
A dynamic connectivity analysis produces a trajectory $C(t)$ or a sequence of matrices $\mathbf W(t)$. The attainable-set concept therefore extends from values to trajectories:
\begin{equation}
\Aset_{C_j|C_i}(c_i(\cdot))=\{C_j(\cdot):C_i(\cdot)=c_i(\cdot),\ x(\cdot)\in\Xspace_{\mathrm{dyn}}\}.
\end{equation}
A static summary such as mean connectivity
\begin{equation}
\bar C=\frac{1}{T}\sum_{t=1}^T C(t)
\end{equation}
does not determine variance, dwell time, transition probabilities, or ordering of states. Many distinct trajectories share the same mean. Dynamic-connectivity reviews emphasize that temporal variability and methodological choices such as windowing or state estimation can substantially affect interpretation \cite{Calhoun2014,Preti2017,Lurie2020,Coelli2025,dePasquale2010,Baker2014,Allen2014}.

\begin{proposition}[Static-to-dynamic non-identifiability]
Suppose $C_1(t)$ and $C_2(t)$ are bounded trajectories with equal time average. Unless the admissible trajectory space imposes additional restrictions, the first two moments, dwell-time distribution, transition structure, and temporal ordering are not uniquely determined by the mean. Hence a static connectivity summary does not identify dynamic connectivity descriptors.
\end{proposition}

This point is especially important when a single dynamic metric is used as a biomarker. ``More variable'' or ``more state-like'' is not a mathematical property of a mean connectivity value; it requires explicit temporal information.

\subsection{Sliding windows as a family of functionals}
A dynamic estimate can be represented by a family of window operators $\{\mathcal W_t\}$ and connectivity functionals,
\begin{equation}
C(t)=F\big(\mathcal W_t[\mathbf y]\big).
\end{equation}
Changing window length, overlap, tapering, or state-estimation method changes the operator $\mathcal W_t$ and hence changes the dynamic functional. A state sequence obtained with a 2-s window is not mathematically the same observable as a state sequence obtained with a 10-s window even when both are described as ``dynamic connectivity.'' The recent EEG/MEG dynamic-connectivity literature emphasizes these methodological degrees of freedom \cite{Coelli2025,Preti2017,Lurie2020,Allen2014,dePasquale2010,Baker2014,Vidaurre2017}.

\begin{proposition}[Mean-preserving dynamic counterexample]
Let $C(t)$ be any bounded scalar connectivity trajectory on $[0,T]$. There exist infinitely many trajectories with the same mean
\begin{equation}
\bar C=T^{-1}\int_0^T C(t)\,dt
\end{equation}
but different variance, autocorrelation, number of state transitions, and dwell-time distributions. Therefore the time-average of a dynamic connectivity measure cannot in general identify those temporal features.
\end{proposition}
\noindent\textit{Construction:} choose any zero-mean perturbation $h(t)$ satisfying the amplitude constraints of the state space and set $C_2(t)=C_1(t)+h(t)$. Then $\bar C_1=\bar C_2$ while temporal second-order structure changes with the choice of $h$.

\subsection{State-based dynamic connectivity}
When dynamic connectivity is summarized by discrete states $S_t\in\{1,\ldots,K\}$, the occupancy vector
\begin{equation}
\boldsymbol\pi=(\pi_1,\ldots,\pi_K)
\end{equation}
does not determine the transition matrix $\mathbf P$ or state ordering. Two processes can have identical occupancy but entirely different transition structures. A review or application that reports only occupancy therefore cannot legitimately infer switching rate or temporal sequencing without additional information. Table~\ref{tab:dynamic} summarizes the corresponding information loss for common dynamic-connectivity summaries.

\begin{table*}[t]
\caption{Dynamic connectivity summaries and information they discard.}
\label{tab:dynamic}
\centering
\scriptsize
\begin{tabularx}{\textwidth}{p{2.7cm} p{3.5cm} X}
\toprule
Summary & Retains & Does not determine \\
\midrule
Time mean & average edge strength & variance, ordering, dwell times, transition probabilities \\
Time variance & dispersion around the mean & ordering, state occupancy, mean unless supplied \\
Occupancy vector & fraction of time in each state & transition matrix and sequence order \\
Transition matrix & one-step switching probabilities & dwell distributions beyond the Markov assumption and state geometry \\
State centroids & representative connectivity patterns & occupancy, transitions, within-state dispersion \\
Entropy of state sequence & temporal unpredictability & specific states or direction of transitions \\
\bottomrule
\end{tabularx}
\end{table*}

\section{Connectivity Matrices and Graph Measures}
\label{sec:graph}
Let $\mathbf W\in\R^{N\times N}$ denote a connectivity matrix. A graph metric is a map $G(\mathbf W)$ from a high-dimensional matrix to a scalar or low-dimensional summary. The compression is severe: an undirected weighted matrix contains $N(N-1)/2$ edge values, whereas mean degree, clustering, efficiency, or modularity can each be a single number \cite{RubinovSporns2010,BullmoreSporns2009,StamReijneveld2007,StamStraaten2012}.

\subsection{Dimension and information loss}
For any scalar graph metric $g(\mathbf W)$, the level set
\begin{equation}
\{\mathbf W:g(\mathbf W)=g_0\}
\end{equation}
is typically high-dimensional. Therefore one graph statistic will generally not identify another. This remains true even when both graph measures are highly correlated on a particular empirical ensemble.

\subsection{Thresholding and density}
Thresholding is itself a state transformation. The compression from the edge matrix to scalar graph summaries is illustrated in Fig.~\ref{fig:graph_compression}, while the broader state-space picture is shown in Fig.~\ref{fig:fiber_shrinkage}.

Thresholding is itself a state transformation:
\begin{equation}
W_{ij}^{(\tau)}=\mathbf 1\{W_{ij}>\tau\}
\end{equation}
for binary thresholding. Since the transformation can discard edge magnitude and change graph density, a graph metric computed after thresholding is not generally determined by the unthresholded metric or vice versa. The dependence of graph results on node choice, density, thresholding, and weight definition is a recurring theme in network neuroscience \cite{RubinovSporns2010,HeEvans2010,Farahani2019,VanWijk2010,Zalesky2010}.

\subsection{Network-theoretic foundations}
Small-worldness, efficiency, modularity, path length, centrality, and community structure are defined on network objects rather than raw time series \cite{WattsStrogatz1998,LatoraMarchiori2001,Blondel2008,Newman2003,BullmoreSporns2009,Sporns2018}. Consequently, cross-measure identifiability becomes a composition problem:
\begin{equation}
C_i(\mathbf Y)\longrightarrow\mathbf W\longrightarrow G(\mathbf W).
\end{equation}
Every stage can discard information. The source connectivity measure need not even uniquely identify the matrix on which a graph measure is subsequently computed.

\begin{figure}[!t]
\centering
\begin{tikzpicture}[node distance=5mm and 9mm,>=Latex,every node/.style={font=\scriptsize}]
\node[draw,rounded corners,minimum width=2.0cm,minimum height=0.7cm] (signals) {EEG signals};
\node[draw,rounded corners,minimum width=2.0cm,minimum height=0.7cm,right=of signals] (matrix) {Connectivity matrix $\mathbf W$};
\node[draw,rounded corners,minimum width=1.7cm,minimum height=0.7cm,right=of matrix] (metric1) {Metric $G_1$};
\node[draw,rounded corners,minimum width=1.7cm,minimum height=0.7cm,below=of metric1] (metric2) {Metric $G_2$};
\draw[->] (signals)--(matrix)--(metric1);
\draw[->] (matrix)--(metric2);
\node[align=center,text width=2.1cm,below=of matrix] {Many matrices can share\\the same scalar metric};
\end{tikzpicture}
\caption{Graph-level compression. A scalar network metric is a many-to-one map of the connectivity matrix. Consequently, one graph summary generally does not identify another graph summary unless the admissible matrix class supplies additional constraints.}
\label{fig:graph_compression}
\end{figure}
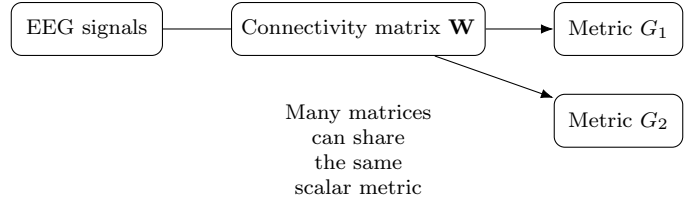

\subsection{Graph measures are second-order summaries of a connectivity matrix}
Let $\mathbf W\in\R^{N\times N}$ be a weighted adjacency matrix derived from a connectivity measure. A graph functional $G(\mathbf W)$ may be global efficiency, modularity, clustering coefficient, characteristic path length, participation coefficient, degree, strength, or another topological summary \cite{RubinovSporns2010,StamReijneveld2007,BullmoreSporns2009,BassettSporns2017,Farahani2019}. Because $G$ maps a matrix into a scalar or a much smaller set of scalars, its fibers are typically high-dimensional.

For example, mean degree is invariant to many redistributions of edge weight. Two weighted networks can therefore have the same mean degree but different modularity and clustering. Likewise, thresholded graphs can have identical density but different connectedness or community structure. A graph summary should consequently be interpreted as an additional compression layer, not as a substitute for the edge matrix.

\begin{proposition}[Graph-summary non-identifiability]
Let $G_1$ and $G_2$ be two nontrivial graph functionals on a class of graphs containing at least two non-isomorphic graphs with equal $G_1$ values but unequal $G_2$ values. Then $G_2$ is not identifiable from $G_1$ over that graph class. In particular, no graph-theoretic identity can be assumed solely because both metrics are calculated from the same connectivity matrix.
\end{proposition}

\subsection{Weighted graphs and thresholding}
Weighted graph measures retain more information than their thresholded counterparts only if the weighting scheme itself is retained. Let
\begin{equation}
W^{(\tau)}_{ij}=\mathbf 1(W_{ij}>\tau).
\end{equation}
Thresholding is many-to-one: every interval of edge weights maps to the same binary value. Therefore a graph metric computed after thresholding can be insensitive to large changes in the underlying continuous connectivity matrix until an edge crosses $\tau$, after which the metric can change discontinuously. Density-preserving threshold schemes add another constraint but still do not restore the lost weight information \cite{VanWijk2010,Zalesky2010}.

\subsection{Network comparison must account for topology and weight scale}
Comparing graph metrics across conditions with different network density can produce artifacts because many metrics depend on node number, density, or weight distribution. The mathematical framework treats density as an explicit conditioning variable when required. A claim that a change in global efficiency is caused by a change in ``connectivity strength'' is therefore incomplete unless the admissible state space controls the weight-to-topology transformation.

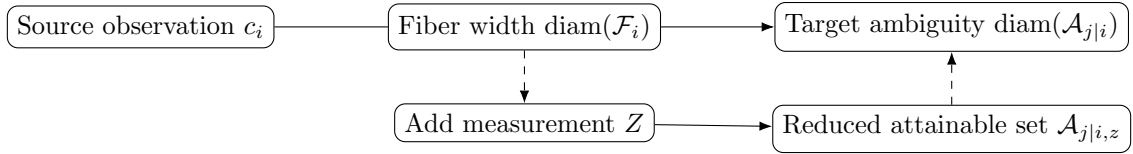
\begin{figure*}[t]
\centering
\begin{tikzpicture}[node distance=7mm and 15mm, >=Latex, every node/.style={font=\small}]
\node[draw,rounded corners,minimum width=3cm,align=center] (a) {Source observation $c_i$};
\node[draw,rounded corners,minimum width=3cm,align=center,right=of a] (b) {Fiber width $\operatorname{diam}(\Fset_i)$};
\node[draw,rounded corners,minimum width=3cm,align=center,right=of b] (c) {Target ambiguity $\operatorname{diam}(\Aset_{j|i})$};
\node[draw,rounded corners,minimum width=3cm,align=center,below=of b] (d) {Add measurement $Z$};
\node[draw,rounded corners,minimum width=3cm,align=center,below=of c] (e) {Reduced attainable set $\Aset_{j|i,z}$};
\draw[->] (a)--(b)--(c);
\draw[->] (d)--(e);
\draw[->,dashed] (b)--(d);
\draw[->,dashed] (e)--(c);
\end{tikzpicture}
\caption{Conceptual relation between hidden state-space width and target ambiguity. A source statistic collapses many admissible states into a fiber. Additional measurements or assumptions reduce that fiber, and the target attainable set can shrink from a broad region to a bounded interval or singleton.}
\label{fig:fiber_shrinkage}
\end{figure*}

\section{Finite-Sample Estimation and Uncertainty}
\label{sec:uncertainty}
The framework distinguishes two layers of uncertainty. First, there is \emph{population identifiability}: whether the target functional is determined by the source functional in the idealized state space. Second, there is \emph{estimation uncertainty}: whether a finite EEG record permits the relevant functionals to be estimated accurately.

Let $\widehat C_i$ and $\widehat C_j$ denote finite-sample estimators. Even if $C_j=g(C_i)$ exactly, generally
\begin{equation}
\widehat C_j\neq g(\widehat C_i)
\end{equation}
due to estimator noise, bias, and differences in effective sample size. Conversely, an empirical near-equality between two estimators may arise from correlated estimation noise even when no population identity exists.

\subsection{Sample size and estimator bias}
Phase metrics can have finite-sample bias, motivating PPC and related corrections \cite{Vinck2010,Vinck2011}. Information-theoretic measures can require substantial sample sizes and are sensitive to the estimator and dimensionality \cite{Kraskov2004,Paninski2003}. Directed models introduce model-order selection and parameter-estimation uncertainty \cite{BarnettSeth2014,BarnettSeth2015}. Reliability should therefore be treated separately from statistical significance. Evidence from large-scale reliability syntheses indicates that functional connectivity estimates can exhibit substantial measurement variability even when group-level effects are detectable \cite{Noble2019,Colclough2016,Holler2017}.

\subsection{Surrogate and permutation inference}
Surrogate testing and permutation procedures are widely used to assess connectivity effects without relying solely on parametric assumptions \cite{Lachaux1999,NicholsHolmes2002,MarisOostenveld2007}. However, significance of a source statistic does not establish identifiability of a target statistic. Null-hypothesis testing and cross-measure inference are logically distinct questions.

\subsection{Multiplicity}
Connectivity analyses often involve many channel pairs, frequencies, conditions, or graph nodes. False discovery rate control and related multiple-comparison procedures remain essential when converting pairwise significance tests into network-level conclusions \cite{BenjaminiHochberg1995}. The attainable-set framework is complementary: multiplicity controls the probability of false positives, whereas identifiability specifies what can be logically inferred from the value of a statistic.

\subsection{Population identifiability versus estimator identifiability}
Let $C_i$ and $C_j$ be population functionals and $\widehat C_i,\widehat C_j$ their finite-sample estimators. Population exactness means
\begin{equation}
C_j=g(C_i)
\end{equation}
on the admissible population state space. This does not imply that the finite-sample estimators satisfy $\widehat C_j=g(\widehat C_i)$ exactly. Bias, variance, regularization, discretization, surrogate choices, and numerical tolerances can all intervene.

Conversely, an empirical regression may produce an accurate estimator of $C_j$ from $\widehat C_i$ on one data distribution even when population exactness fails. Such a regression is a predictive model, not a proof of identifiability.

\subsection{Estimator uncertainty as a second attainable set}
Suppose the estimator can be represented as
\begin{equation}
\widehat C_i=C_i(x)+\epsilon_i,
\qquad
\widehat C_j=C_j(x)+\epsilon_j,
\end{equation}
with an uncertainty model for $(\epsilon_i,\epsilon_j)$. The observed-source event $\widehat C_i=c$ no longer corresponds to a single population fiber but to a weighted neighborhood of fibers. A practical uncertainty-aware attainable set can therefore be written schematically as
\begin{equation}
Q_{j|i}(u;c)=\Pr(\widehat C_j\leq u\mid\widehat C_i=c,\ x\in\Xspace),
\end{equation}
followed by the uncertainty-aware set
\begin{equation}
\Aset^{\mathrm{est}}_{j|i}(c;\alpha)
=\{u:Q_{j|i}(u;c)\text{ meets the chosen }\alpha\text{-criterion}\}.
\end{equation}
where the exact construction depends on the sampling model. The purpose is not to impose a single universal inferential recipe, but to make explicit that finite-sample uncertainty wraps around the logical identifiability problem.

\subsection{Permutation and surrogate inference}
Connectivity analyses frequently use surrogate or permutation procedures. Such methods are valuable for testing null hypotheses and controlling false-positive rates \cite{MarisOostenveld2007,NicholsHolmes2002,Winkler2014}. However, a significance test answers whether an observed statistic is unusual under a null; it does not answer whether a second connectivity functional is uniquely determined by the observed statistic. Identifiability and hypothesis testing should therefore be reported separately.

\subsection{Multiplicity and selection}
Connectivity studies often evaluate many channel pairs, frequencies, conditions, and graph metrics. The number of statistical tests can become very large, motivating error-control procedures such as the false discovery rate \cite{BenjaminiHochberg1995,BenjaminiYekutieli2001,Storey2002}. Multiple-testing correction changes the inferential threshold, not the underlying mapping $C_i\mapsto C_j$. A statistically significant source-target association can still be non-identifying, while an identifying mathematical relation may hold regardless of the multiplicity procedure.

\subsection{Reliability is orthogonal to identifiability}
A connectivity measure can be highly reliable yet weakly informative about another measure, or noisy yet theoretically sufficient under a strict model. Test--retest reliability and identifiability therefore describe different properties \cite{Noble2019,Holler2017,Colclough2016}. Reliability asks whether the same functional is reproducible; identifiability asks whether one functional determines another over the admissible state space. Table~\ref{tab:uncertainty} separates these layers from one another.

\begin{table*}[t]
\caption{Three distinct sources of uncertainty in cross-connectivity inference.}
\label{tab:uncertainty}
\centering
\scriptsize
\begin{tabularx}{\textwidth}{p{2.6cm} p{4.0cm} X}
\toprule
Layer & Question & Example consequence \\
\midrule
Population non-identifiability & Which target values are compatible with the source value? & non-singleton attainable set even with infinite data \\
Sampling uncertainty & How far can an estimator deviate from its population functional? & confidence intervals, bias, sample-size dependence \\
Pipeline uncertainty & How much does the mapping change with reference/filter/source choices? & different attainable sets for different preprocessing pipelines \\
\bottomrule
\end{tabularx}
\end{table*}

\section{Worked Analytical Case Studies}
\label{sec:cases}

\subsection{Case 1: covariance $\rightarrow$ correlation}
Fix $c=\Cov(X,Y)$ and known $\sigma_X,\sigma_Y$. The attainable set is the singleton
\begin{equation}
\Aset_{r|\mathrm{cov}}(c)=\left\{\frac{c}{\sigma_X\sigma_Y}\right\}.
\end{equation}
If the scales are not fixed but satisfy $\sigma_X\in[a_1,a_2]$ and $\sigma_Y\in[b_1,b_2]$, with positive bounds, then
\begin{equation}
\Aset_{r|\mathrm{cov}}(c)
=\left[\frac{c}{a_2b_2},\frac{c}{a_1b_1}\right]
\end{equation}
for $c\geq0$, with the endpoints reversed when $c<0$. The interval is an analytical bound induced by explicit scale bounds.

\subsection{Case 2: correlation $\rightarrow$ Gaussian MI}
For $\rho\in(-1,1)$ under the bivariate Gaussian model,
\begin{equation}
\Aset_{\mathrm{MI}|r}(\rho)=\left\{-\frac12\log(1-\rho^2)\right\}.
\end{equation}
Thus the attainable set collapses to one point only because the admissible distribution class has been restricted.

\subsection{Case 3: PLV $\rightarrow$ coherence}
For constant amplitudes, the attainable set is a singleton at the population level. With variable amplitudes, two phase distributions can yield the same PLV while amplitude-phase covariance changes the coherency magnitude. Hence the unrestricted attainable set is generally multi-valued.

\subsection{Case 4: PLI $\rightarrow$ wPLI}
The constructive example in Section~VI gives
\begin{equation}
\begin{aligned}
\PLV_1&=\PLV_2=0.5,\\
\wPLI_1&=0.5,\qquad \wPLI_2\approx0.935.
\end{aligned}
\end{equation}
Thus the attainable set of wPLI conditioned on PLI $=0.5$ contains at least two values. This is a proof of non-identifiability, not merely an intuition.

\subsection{Case 5: mean dynamic FC $\rightarrow$ variability}
Consider two binary-state trajectories with equal dwell times but different ordering, for example $0101\ldots$ and $0011\ldots$. Their means can be equal while their autocorrelation, transition count, and dwell-time statistics differ. A static mean therefore cannot determine temporal organization without further constraints.

\subsection{Case 6: sensor connectivity $\rightarrow$ source connectivity}
Let $\mathbf y=\mathbf L\mathbf s$ with full-rank or rank-deficient $\mathbf L$. Unless $\mathbf L$ and an appropriate inverse are known and uniquely identified, many source configurations can lead to the same sensor data, and source connectivity is consequently not determined by sensor connectivity alone. This is a direct consequence of the inverse problem \cite{Michel2004,Grech2008,PascualMarqui1994}.

\subsection{Case study: correlation to mutual information under three state spaces}
Consider a fixed correlation value $\rho_0$. Define three nested state spaces: $\Xspace_1$ containing jointly Gaussian pairs; $\Xspace_2$ containing a broader class of elliptically symmetric distributions; and $\Xspace_3$ containing arbitrary continuous bivariate distributions satisfying finite second moments. On $\Xspace_1$, mutual information is exactly
\begin{equation}
\MI=-\tfrac12\log(1-\rho_0^2).
\end{equation}
On larger spaces, the fiber generally contains multiple joint distributions with the same $\rho_0$ but different mutual information. The important result is not simply ``MI and correlation are different.'' It is that the strength of inference depends monotonically on the assumptions defining $\Xspace$.

\subsection{Case study: PLI and wPLI}
The constructive example in Section~VI proves that the PLI value $0.50$ can coexist with wPLI values near $0.50$ and $0.94$. Hence a reported PLI value should not be converted to a unique wPLI value. A bounded conversion may still be possible if the imaginary cross-spectrum magnitude is known to lie in a constrained range. Let $0<m_{\min}\leq |\Im Z|\leq m_{\max}$. Then the sign-weighted mean is bounded by linear-fractional optimization over the sign probabilities and magnitude interval. The precise bound depends on the additional constraints; the framework deliberately avoids asserting a universal formula where none has been proved.

\subsection{Case study: static connectivity to a dynamic trajectory}
Suppose only $\bar C=0.4$ is reported for a dynamic edge. One admissible trajectory can be constant at $0.4$. Another can alternate between $0.1$ and $0.7$ for equal durations. Both have identical means, but the second has nonzero variance and repeated transitions. A third trajectory can remain at $0.7$ for 40\% of the time and at $0.2$ for 60\%, again matching the same mean. Without additional temporal constraints, the mean therefore identifies neither variance nor transition pattern.

\subsection{Case study: connectivity matrix to graph modularity}
Let $\mathbf W$ be a weighted matrix. A scalar modularity score $Q$ defines a high-dimensional level set. Redistributing edge weights within communities can preserve $Q$ while altering individual edge strengths. Conversely, two matrices can have similar average edge weights but markedly different modularity because the placement of those weights changes community structure. Therefore an edge-level measure cannot in general be reconstructed from a global modularity value and vice versa. The inference direction has to be specified before any ``network conversion'' is discussed.

\subsection{Case study: sensor reference as conditioning information}
Suppose a study reports a correlation or coherence computed after an average reference. Reusing that value as if it were obtained from a linked-mastoid or REST-style reference changes the problem because the reference operator is part of the observable. In the proposed notation, one should compare $C_i^{(R_1)}$ with $C_j^{(R_1)}$, or explicitly model the transformation from $R_1$ to $R_2$. Treating $R$ as a nuisance-free label silently assumes that the connectivity functional is invariant to reference transformation, an assumption that does not generally hold \cite{Nunez1997,Nunez1997,Yao2001,Trujillo2017,Rios2019}.

\subsection{Case study: source leakage correction}
Consider source signals $\mathbf s$ reconstructed from sensors and a leakage correction operator $\mathcal O$. A source connectivity estimate based on $C(\mathbf s)$ and one based on $C(\mathcal O\mathbf s)$ answer different questions. If the scientific target is a coupling measure designed to suppress zero-lag mixing, it is mathematically appropriate to define that target as $C\circ\mathcal O$. It is not appropriate to compare the corrected result with an uncorrected result as though the two were noisy estimates of a single latent scalar. Leakage-correction studies demonstrate that correction choices also interact with reliability \cite{Colclough2015,Colclough2016,Brookes2012,SanchezBornot2018}.

\subsection{Case study: conditional directed connectivity}
Suppose pairwise GC from $X$ to $Y$ is reported. A third process $Z$ can serve as a common driver or mediator. Conditional GC incorporates $Z$ and can therefore differ from pairwise GC. The source pairwise GC is not sufficient to reconstruct the conditional result because the omitted interactions with $Z$ are hidden coordinates of the state. This is the directed analogue of the marginal-to-partial correlation problem and motivates explicit multivariate modeling \cite{Geweke1982,KaminskiBlinowska2001,BarnettSeth2014,BarnettSeth2015}.

\section{A Unified Decision Framework for Cross-Measure Inference}
\label{sec:decision}
A practical cross-measure question can be organized into the following sequence. First, define the source and target functionals exactly. Second, define the admissible state space. Third, identify variables discarded by the source measure but required by the target. Fourth, determine whether the target is constant over each source fiber. If not, impose realistic bounds or supply a minimal augmentation. Finally, distinguish population-level identifiability from finite-sample estimation. Table~\ref{tab:decision} provides a compact checklist for these steps, and Fig.~\ref{fig:decision} gives the corresponding flow.

\begin{table*}[t]
\caption{Decision framework for a claim that one EEG connectivity measure determines another.}
\label{tab:decision}
\centering
\scriptsize
\begin{tabularx}{\textwidth}{p{1.1cm} p{3.2cm} p{5.0cm} X}
\toprule
Step & Question & Mathematical test & Interpretation \\
\midrule
1 & What are the exact definitions? & Specify $C_i$ and $C_j$ including normalization and conditioning & Prevents accidental comparison of non-equivalent variants \\
2 & What is admissible? & Define $\Xspace$ and nuisance parameter ranges & Makes assumptions visible \\
3 & What does the source discard? & Characterize fibers $\Fset_i(c)$ & Reveals hidden degrees of freedom \\
4 & Is there a factorization? & Test whether $C_j=g\circ C_i$ & Singleton attainable sets imply exact inference \\
5 & If not, can bounds be derived? & Optimize $C_j$ over $C_i=c$ and $x\in\Xspace$ & Produces analytical intervals or regions \\
6 & What extra information is sufficient? & Find minimal augmentation $Z$ shrinking fibers to singletons & Converts an empirical caveat into a design requirement \\
7 & What is the estimator uncertainty? & Propagate sampling distributions of $\widehat C_i,\widehat C_j$ & Separates identifiability from estimation \\
8 & Does the claim survive perturbations? & Repeat over references, filters, windows, leakage models & Tests methodological robustness \\
\bottomrule
\end{tabularx}
\end{table*}

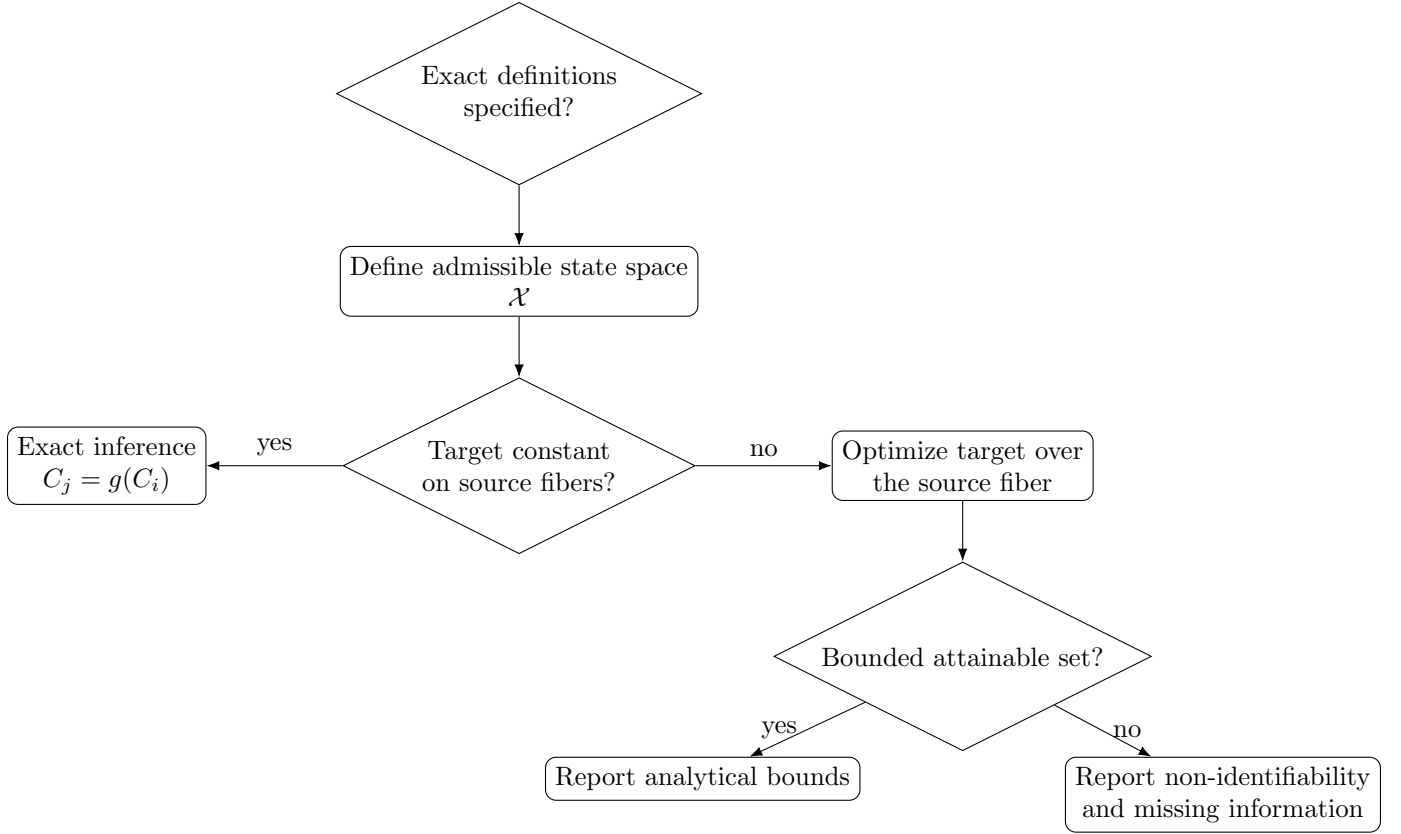
\begin{figure*}[t]
\centering
\begin{tikzpicture}[node distance=8mm and 12mm,>=Latex,every node/.style={font=\small}]
\node[draw,diamond,aspect=2,align=center] (def) {Exact definitions\\specified?};
\node[draw,rounded corners,align=center,below=of def] (state) {Define admissible state space\\$\Xspace$};
\node[draw,diamond,aspect=2,align=center,below=of state] (factor) {Target constant\\on source fibers?};
\node[draw,rounded corners,align=center,left=18mm of factor] (exact) {Exact inference\\$C_j=g(C_i)$};
\node[draw,rounded corners,align=center,right=18mm of factor] (bound) {Optimize target over\\the source fiber};
\node[draw,diamond,aspect=2,align=center,below=of bound] (bounded) {Bounded attainable set?};
\node[draw,rounded corners,align=center,below left=7mm and 1mm of bounded] (boundedout) {Report analytical bounds};
\node[draw,rounded corners,align=center,below right=7mm and 1mm of bounded] (nonid) {Report non-identifiability\\and missing information};
\draw[->] (def)--(state);
\draw[->] (state)--(factor);
\draw[->] (factor)--node[above]{yes}(exact);
\draw[->] (factor)--node[above]{no}(bound);
\draw[->] (bound)--(bounded);
\draw[->] (bounded)--node[left]{yes}(boundedout);
\draw[->] (bounded)--node[right]{no}(nonid);
\end{tikzpicture}
\caption{Decision flow for cross-connectivity inference. The framework favors exact factorization when available, analytical bounds when the target attainable set is bounded, and an explicit non-identifiability statement when neither exactness nor useful bounds follow from the specified assumptions.}
\label{fig:decision}
\end{figure*}

\subsection{A reusable claim template}
For a manuscript reporting a cross-measure relationship, the following mathematical sentence is sufficient to make the logical strength explicit:
\begin{quote}
Let $\Xspace_0$ denote the admissible state space defined by [assumptions]. For source measure $C_i$ and target measure $C_j$, the attainable set $\Aset_{j|i}(c_i)$ is [singleton / bounded non-singleton / unbounded]. Therefore $C_j$ is [exactly determined / conditionally determined given $Z$ / bounded by $L$ and $U$ / not identifiable] from $C_i$ over $\Xspace_0$.
\end{quote}
This statement is compact enough for application papers while retaining the mathematical assumptions that determine whether a conversion is valid.

\subsection{Study-design implication: measure what closes the fiber}
Once a non-identifiable source-target pair has been identified, the framework suggests a principled way to select additional measurements. Rather than adding variables opportunistically, investigators can ask which candidate variable $Z_k$ produces the greatest expected reduction in attainable-set diameter,
\begin{equation}
\Delta_k=\mathbb E\left[\operatorname{diam}(\Aset_{j|i}(C_i))-\operatorname{diam}(\Aset_{j|i,Z_k}(C_i,Z_k))\right].
\end{equation}
The quantity can be estimated under a chosen probabilistic design or evaluated analytically for restricted signal families. It is not a universal objective function, but it provides a mathematical language for experiment design: choose measurements that close the degrees of freedom most responsible for the target ambiguity.

\section{Critical Synthesis of the EEG Connectivity Literature}
\label{sec:synthesis}

\subsection{Taxonomy alone is insufficient}
A conventional taxonomy tells the reader that coherence is spectral, PLV is phase-based, MI is information-theoretic, and graph measures are network-level. The present framework asks a stronger question: what information survives each transformation, and can the surviving information support an inference about a second transformation? This shift changes the unit of comparison from the name of the measure to its equivalence classes on state space.

\subsection{Universal relationships versus empirical relationships}
The literature contains at least four logically different kinds of relationships. The first is a universal mathematical identity valid over the declared state space. The second is a conditional identity valid after distributional or signal-model assumptions are imposed. The third is a statistical association observed on a dataset or benchmark. The fourth is a predictive mapping learned from data. Confusing these classes is a major source of overstatement. This hierarchy is summarized in Table~\ref{tab:claimhierarchy}.

\begin{table*}[t]
\caption{Hierarchy of cross-measure claims.}
\label{tab:claimhierarchy}
\centering
\scriptsize
\begin{tabularx}{\textwidth}{p{2.5cm} p{3.2cm} X X}
\toprule
Claim type & Formal statement & Evidence required & Typical wording \\
\midrule
Universal identity & $C_j=g(C_i)$ on all admissible states & mathematical derivation & ``is determined by'' \\
Conditional identity & $C_j=g(C_i;z)$ under assumptions & derivation + explicit state restrictions & ``under the Gaussian / constant-amplitude model'' \\
Bounded inference & $L(C_i)\leq C_j\leq U(C_i)$ & analytical optimization or proof of bounds & ``is bounded by'' \\
Empirical association & $C_j\approx g(C_i)$ on a sample & validation dataset + uncertainty analysis & ``is associated with'' \\
Predictive model & $\widehat C_j=\widehat g(C_i,\mathbf z)$ & out-of-sample validation & ``can predict'' \\
Non-identifiability & $\exists x_1,x_2:C_i(x_1)=C_i(x_2),\ C_j(x_1)\neq C_j(x_2)$ & counterexample or construction & ``cannot be uniquely inferred'' \\
\bottomrule
\end{tabularx}
\end{table*}

\begin{table*}[t]
\caption{Evidence layers represented in the connectivity literature and their role in the present framework.}
\label{tab:evidence_map}
\centering
\scriptsize
\begin{tabularx}{\textwidth}{p{2.6cm} p{5.5cm} X}
\toprule
Evidence layer & Representative literature & Role in the mathematical framework \\
\midrule
Definitions and conceptual foundations & Friston (functional/effective connectivity), Horwitz, Varela et al. \cite{Friston1994,Horwitz2003,Varela2001} & establishes the distinction between statistical dependence and directed/biological interpretations \\
Connectivity tutorials and reviews & Pereda et al., Sakkalis et al., Bastos and Schoffelen, Chiarion et al. \cite{Pereda2005,Sakkalis2011,BastosSchoffelen2016,Chiarion2023} & defines estimator families and recurring interpretational issues \\
Phase synchronization & Lachaux et al., Stam et al., Vinck et al., Aydore et al., Bruña et al. \cite{Lachaux1999,Stam2007,Vinck2010,Vinck2011,Aydore2013,Bruna2018} & supplies the phase-based functionals used for constructive identifiability analysis \\
Spectral connectivity & Nunez et al., Nolte et al., Schoffelen and Gross, Pascual-Marqui et al. \cite{Nunez1997,Nunez1997,Nolte2004,SchoffelenGross2009,PascualMarqui2011} & formalizes frequency-domain dependence and source-space considerations \\
Directed connectivity & Granger, Geweke, Kamiński and Blinowska, Baccalá and Sameshima, Barnett and Seth, Vicente et al. \cite{Granger1969,Geweke1982,KaminskiBlinowska2001,BaccalaSameshima2001,BarnettSeth2014,BarnettSeth2015,Vicente2011} & establishes model-conditional relationships and their assumptions \\
Cross-frequency coupling & Canolty, Canolty and Knight, Tort, Aru, Hyafil \cite{Canolty2006,CanoltyKnight2010,Tort2008,Tort2010,Aru2015,Hyafil2017} & motivates waveform- and estimator-dependent attainable sets \\
Leakage, reference, and preprocessing & Yao, Nolte, Brookes, Colclough, Rios, Widmann, Bigdely-Shamlo, Pion-Tonachini \cite{Yao2001,Nolte2004,Brookes2012,Colclough2015,Rios2019,Widmann2015,BigdelyShamlo2015,PionTonachini2019} & places processing operators inside the observation map \\
Dynamic connectivity & de Pasquale, Allen, Preti, Calhoun, Vidaurre, Coelli \cite{dePasquale2010,Allen2014,Preti2017,Calhoun2014,Vidaurre2017,Coelli2025} & extends fibers from scalar observations to trajectories and state sequences \\
Network science & Watts--Strogatz, Newman, Rubinov--Sporns, Bassett--Sporns, Sporns \cite{WattsStrogatz1998,Newman2003,RubinovSporns2010,BassettSporns2017,Sporns2018} & motivates the matrix-to-graph compression problem \\
Reliability and statistical inference & Holler et al., Noble et al., Maris and Oostenveld, Nichols and Holmes, Benjamini and Hochberg \cite{Holler2017,Noble2019,MarisOostenveld2007,NicholsHolmes2002,BenjaminiHochberg1995} & separates reproducibility and significance from population identifiability \\
\bottomrule
\end{tabularx}
\end{table*}

\subsection{Sensor space versus source space}
Sensor-space and source-space connectivity should be understood as different observables linked by a measurement model. Source reconstruction can improve interpretability and spatial specificity, but the inverse problem is not automatically unique \cite{Michel2004,Grech2008,PascualMarqui2011,SchoffelenGross2009}. Leakage correction further changes the target functional. Therefore a statement of the form ``source connectivity equals underlying connectivity'' exceeds what the forward model alone can establish.

\subsection{Static versus dynamic connectivity}
The dynamic literature has introduced rapidly evolving descriptors, including time-varying matrices, state sequences, and transient network models \cite{Preti2017,Coelli2025,Baker2014}. The attainable-set framework indicates that every dynamic summary should be viewed as a functional of a trajectory. A static edge average, dynamic variance, and state dwell-time statistic are different compressions of the same trajectory and should not be treated as interchangeable.

\subsection{Graph metrics and downstream machine learning}
Graph measures and machine-learning features further compress connectivity matrices. This can be scientifically useful but creates another layer of non-identifiability. A classifier trained on one connectivity representation does not establish that another representation contains equivalent information. Recent EEG connectome perspectives underscore the expanding role of network representations in cognitive and clinical neuroscience \cite{ZhangChen2025}. The present framework suggests that transfer between representations should first be assessed as a question of information preservation. The evidence layers used in this review are mapped in Table~\ref{tab:evidence_map}.

\section{Reporting Standards Suggested by the Framework}
\label{sec:reporting}
A connectivity study intended to support a cross-measure inference should report enough information to reconstruct the admissible state space. Table~\ref{tab:reporting} gives the minimum reporting fields proposed by the framework.

\begin{table*}[t]
\caption{Minimum reporting checklist for connectivity studies used in cross-measure comparisons.}
\label{tab:reporting}
\centering
\scriptsize
\begin{tabularx}{\textwidth}{p{3.0cm} X X}
\toprule
Domain & Minimum information & Why it matters for identifiability \\
\midrule
Acquisition & channel count, montage, sampling rate, hardware, recording duration & defines observable dimension and measurement operator \\
Reference & explicit reference scheme and any re-referencing & can change covariance, information, and network topology \\
Filtering & passbands, filter type/order, phase mode, edge handling & changes spectra and phase \\
Artifact processing & ICA/BSS/ASR/SSP method, component selection rules & changes the signal state before connectivity estimation \\
Epoching & window length, overlap, stationarity assumptions & controls estimator variance and dynamic interpretation \\
Connectivity estimator & exact equation, normalization, software implementation & avoids mixing variants with different mappings \\
Frequency analysis & band limits, spectral estimator, tapering, smoothing & affects frequency-domain functionals \\
Source analysis & forward model, inverse method, parcellation, leakage correction & determines source-space observable \\
Statistics & estimator uncertainty, surrogate/permutation design, multiplicity correction & separates inference from identifiability \\
Cross-measure claim & assumptions, conditioning variables, validity range & makes exactness versus approximation explicit \\
\bottomrule
\end{tabularx}
\end{table*}

\subsection{Recommended cross-measure report fields}
For a study whose goal is to compare or convert two connectivity measures, the following fields should be treated as mathematical inputs rather than optional metadata:
\begin{enumerate}[leftmargin=*]
\item the exact population/sample formula for both measures, including normalization and bias correction;
\item the signal representation (sensor, source, parcel, component, band-limited analytic signal, or envelope);
\item reference and montage operators;
\item filter type, cutoff frequencies, order, phase mode, padding, and edge handling;
\item epoch length, overlap, number of independent segments, and stationarity assumptions;
\item source model, inverse method, regularization, parcellation, and leakage correction;
\item conditioning variables for multivariate/directed measures;
\item surrogate generation and statistical correction;
\item the admissible state assumptions used to justify any claimed conversion or bound;
\item whether the statement is intended as an identity, conditional identity, bound, empirical association, or prediction model.
\end{enumerate}

The point is not to burden every EEG paper with abstract mathematics. Instead, these fields provide a minimum specification whenever a result from one connectivity family is used to interpret another. A short paragraph can often make the distinction explicit.

\subsection{A reproducibility principle}
The mathematical object being compared should be reconstructable from the report. A citation to a method name alone is insufficient when multiple implementations differ in normalization, regularization, windowing, or conditioning. A reproducibility-oriented report should therefore include the equations and implementation choices needed to reconstruct $C_i$ and $C_j$ from the same input state \cite{Sakkalis2011,DelormeMakeig2004,Gramfort2014}.

\section{Limitations of the Present Review}
\label{sec:limitations}
The framework is deliberately general, but several limitations should be stated. First, the attainable set is only as meaningful as the admissible state space. If the state restrictions are unrealistically broad, the bounds may become too loose to be useful; if they are unrealistically narrow, an apparent exact identity may be an artifact of model assumptions. Second, many clinically used connectivity estimators are implemented in software packages with undocumented defaults or subtle finite-sample corrections. The manuscript therefore emphasizes functional definitions rather than claiming that a single software implementation represents an entire method family. Third, constructive non-identifiability proofs usually establish impossibility over a model class rather than over every conceivable biophysical generative model. Their conclusions should be matched to the state space under discussion. Fourth, the review is mathematical and methodological rather than a meta-analysis of effect sizes across clinical populations. Application-specific clinical conclusions remain dependent on the original empirical studies. Finally, the bibliography is intentionally curated rather than mechanically padded: the aim is to cite traceable sources that support the mathematical claims and major methodological families, not to create an arbitrary numerical reference target.

\section{Open Mathematical Problems and Future Directions}
\label{sec:future}
Several mathematical extensions follow naturally.

First, many attainable-set calculations are high-dimensional constrained optimization problems. Convex relaxations, moment methods, semidefinite programming, and optimal-transport formulations may provide scalable bounds when the state constraints admit tractable structure.

Second, finite-sample identifiability should be studied jointly with population identifiability. The population attainable set describes logical inference; estimator distributions describe practical uncertainty. Their combination could yield probabilistic attainable sets or confidence envelopes for cross-measure inference.

Third, the framework should be extended to multivariate and higher-order interactions. Multivariate mutual information, conditional dependence, hypergraph representations, and higher-order phase interactions introduce richer fibers than pairwise measures.

Fourth, dynamic connectivity can be formalized over path spaces rather than scalar intervals. State-space and optimal-transport methods may be useful for characterizing which temporal features are implied by one dynamic summary and which are not.

Fifth, inverse uncertainty could be integrated directly. Rather than conditioning on a single lead field or source-reconstruction solution, the admissible state space could include an ensemble of plausible forward and inverse models. The resulting target attainable set would quantify methodological uncertainty as part of the scientific inference.

Sixth, standardized benchmark suites could be constructed with analytically controllable signal families. A benchmark could score connectivity methods by the structure of their attainable sets, estimator robustness, sensitivity to known nuisances, and cross-measure information preservation, rather than by downstream classification performance alone.

Seventh, the same framework may be useful beyond EEG. MEG, intracranial electrophysiology, calcium imaging, and multimodal fusion all involve transformations from latent activity to statistical summaries. The mathematical question of what one summary identifies about another is therefore broader than any one modality.

\section{Conclusion}
\label{sec:conclusion}
EEG functional connectivity is not one observable but a family of statistical functionals that compress different aspects of neural signals. Consequently, the value of one connectivity measure generally does not determine another. This paper formalized that issue by defining the cross-connectivity attainable set: the set of all target values compatible with a reported source value over an explicitly stated admissible EEG state space.

The framework distinguishes exact determination, conditional determination, analytical bounded inference, approximate identification, and non-identifiability. These distinctions were illustrated analytically for covariance and correlation, Gaussian mutual information and correlation, coherence and PLV, PLI and wPLI, Granger causality and transfer entropy, dynamic connectivity summaries, and graph metrics. The analysis also demonstrates why reference choice, filtering, volume conduction, source leakage, inverse modeling, and thresholding belong inside the mathematical formulation rather than being treated solely as implementation details.

The principal practical implication is straightforward. A cross-measure statement should always answer three questions: what state information was retained by the source measure; what information was discarded; and which assumptions, bounds, or additional measurements make the target recoverable. When an exact factorization exists, it should be stated as such. When only a conditional relationship exists, the conditions should be written explicitly. When the attainable set is bounded but multi-valued, analytical bounds are more defensible than scalar conversion. When no useful bound exists, non-identifiability is the scientifically correct conclusion.

The framework therefore does not seek to identify a universally superior connectivity estimator. Instead, it provides a common mathematical language for comparing estimators by the partitions they induce on signal space, the information they preserve, and the strength of the cross-measure inferences they support. This language becomes increasingly relevant as EEG research moves toward dynamic, multivariate, source-resolved, and graph-based connectomics.

\section*{Acknowledgment}
The author thanks colleagues and students whose discussions of EEG signal processing, connectivity analysis, and mathematical modeling informed the development of this framework.

\section*{Conflict of Interest}
The author declares no conflict of interest related to this review.

\section*{Data Availability}
No new human-subject data were collected for this article. The manuscript is a mathematical and literature-based review.

\section*{Funding}
No external funding was received specifically for the preparation of this article.

\bibliographystyle{IEEEtran}
\bibliography{ieeeaccess_eeg_fc_refs}

\end{document}